\pdfoutput=1
\documentclass[11pt]{article}
\usepackage[margin=1in]{geometry}
\usepackage{amsmath,amssymb,amsthm}
\usepackage{graphicx}
\usepackage{hyperref}
\usepackage{booktabs}
\usepackage{tikz}
\usetikzlibrary{positioning,arrows.meta,shapes,decorations.pathreplacing}
\usepackage{xcolor}
\usepackage{enumitem}
\usepackage[protrusion=true, expansion=false]{microtype}

\newcommand{\Aut}{\mathrm{Aut}}
\newcommand{\bZ}{\mathbb{Z}}
\newcommand{\Trow}{T_{\mathrm{row}}}
\newcommand{\Tcol}{T_{\mathrm{col}}}
\newcommand{\beq}{\beta_1}
\newcommand{\calT}{\mathcal{T}}
\newcommand{\swp}{\mathrm{swap}}
\newcommand{\rev}{\mathrm{rev}}
\newcommand{\rk}{\mathrm{rank}}

\newtheorem{theorem}{Theorem}
\newtheorem{lemma}[theorem]{Lemma}
\newtheorem{proposition}[theorem]{Proposition}
\newtheorem{corollary}[theorem]{Corollary}
\theoremstyle{definition}
\newtheorem{definition}[theorem]{Definition}
\theoremstyle{remark}
\newtheorem{remark}[theorem]{Remark}

\begin{document}

\title{Which Wallpaper Groups Arise from Tiled Games?}
\author{Matthew Fried\\
Farmingdale State College, SUNY\\
\texttt{friedm1@farmingdale.edu}}
\date{}

\maketitle

\begin{abstract}
Which discrete symmetry groups can arise from strategic interaction?
We tile the plane with copies of a bimatrix game's support complex,
joined by controlled boundary rules, and show that all seventeen
wallpaper groups act on the resulting covers: explicit generators,
each a machine-verified graph automorphism, every realization
certified as the exact toroidal quotient, with types identified by a
crystallographic recognizer in exact rational arithmetic and
cross-validated in GAP. A three-line lemma turns the classical
symmorphic/non-symmorphic distinction into a lattice classification:
realizations whose translations contain the full tile lattice exist
precisely for the thirteen symmorphic groups, and the four
non-symmorphic groups are realized at translation-lattice index
exactly two, the minimum possible: the tile is the glide's
half-step.

Two computational tracks accompany the construction. On the graph
track, quotienting a straight cover by its translations recovers the
tile exactly, $\beq(M/\calT)=\beq(K)$, and swap boundaries add
exactly $\binom m2$, independent of payoffs and of cover size. On
the game track, detecting a duplicated-strategy cover is a
linear-time payoff scan, one tile solution folds to a full
translation orbit of cover equilibria, and the tiled
correlated-equilibrium system has dimension exactly $r(d-q)+q$, with
expansion impossible. The polymatrix cover then carries the symmetry
outright: every wallpaper action, glides included, is a group of
genuine game automorphisms, equilibria collapse along any symmetry
subgroup to a folded fixed-point problem, and a decorated refinement
has game automorphism group exactly the toroidal wallpaper group.
\end{abstract}

\section{Introduction}
\label{sec:intro}

\subsection{The Realization Question}

A bimatrix game $(A,B) \in \mathbb{R}^{m\times n}\times\mathbb{R}^{m\times n}$
has a natural combinatorial object associated with it: the
\emph{support complex} $K(A,B)$, whose nodes are candidate strategy
support pairs $(S_i, S_j)$ connected by single-edit
(pivot-adjacency) edges. This complex is the Cartesian product of
two per-player support-edit graphs, and its first Betti number
$\beq(K)$, the dimension of its cycle space, is a payoff-free
invariant computed in closed form in Section~\ref{sec:prelim}; for
generic payoffs the equilibrium-feasible support pairs form a
discrete independent set inside $K$
(Appendix~\ref{app:tile}), so every $\beq$ statement in this paper
concerns the ambient arena and its symmetries. Separately, computing
a Nash equilibrium of a bimatrix game, for which Lemke and
Howson~\cite{lh1964} gave the classical pivoting algorithm, is
PPAD-complete~\cite{daskalakis2009,chen2009}.

The forward problem is familiar: given a game, find its symmetries
and exploit them. This paper studies the inverse problem:
\emph{which symmetry groups are realizable by game-theoretic
interaction at all?} We answer it for periodic planar interaction.
Tiling copies of $K(A,B)$ into a multigame cover $M_{r,s}$, with
boundary rules deciding how adjacent copies communicate, produces
graphs on which wallpaper groups, the 17 crystallographic symmetry
groups of the plane, act by automorphisms. The realization theory
splits along the classical symmorphic/non-symmorphic line (a
wallpaper group is \emph{symmorphic} when an origin can be chosen so
that every symmetry is a rotation or reflection followed by a
lattice translation; thirteen of the seventeen are, four are not),
and the split is not an accident of our constructions but a theorem
about
integral affine actions (Lemma~\ref{lem:symmorphic}): a wallpaper
group acting by affine maps on the tile grid, with translation
subgroup containing the full grid, must be symmorphic. All thirteen
symmorphic groups are realized on plain straight covers with
verified graph automorphisms (Theorem~\ref{thm:main}). A
non-symmorphic group must instead place its own translation lattice
at proper index inside the tile grid, so that the tile itself serves
as the glide's half-step; this mechanism realizes all four
non-symmorphic groups, each certified (Theorem~\ref{thm:main}(ii)).
All seventeen wallpaper groups therefore arise, and
Theorem~\ref{thm:lattice} pins the classical line exactly: a
realization whose translations contain the full tile lattice exists
if and only if the group is symmorphic. This is a complete realization result for the crystallographic
geometry of tiled equilibrium search spaces.

\subsection{Symmetry That No Game Possesses}

The non-symmorphic groups ($pg$, $pmg$, $pgg$, $p4g$) describe
symmetry that cannot exist locally. First, the game meaning of the group, in plain terms. An element of
a wallpaper group acting on a tiled game is a relabeling of who is
playing where that changes nothing strategic: a translation says the
interaction one tile over is the same interaction; a rotation says
the arena has no preferred direction; a mirror says the left-handed
and right-handed layouts are the same game; and a glide reflection
says the pattern repeats only after a shift combined with a flip, a
symmetry no single tile can exhibit. Wherever the group acts by
genuine game automorphisms (Section~\ref{sec:poly}), profiles
related by such a relabeling have identical payoffs, so equilibria
come in orbits and symmetric equilibria solve a folded problem.

A glide reflection composes a
translation with a reflection; its square is a nontrivial
translation, so it has infinite order and unbounded orbits on the
Euclidean plane, and no bounded planar object, in particular no
single tile, is invariant under one. Finite toroidal quotients
inherit finite glide-induced automorphisms, and
Section~\ref{sec:poly} realizes these as genuine game
automorphisms. In a symmorphic group one can choose an origin so that
every symmetry is a point-group element followed by a lattice
translation; in the four non-symmorphic groups no such origin
exists, so glides are unavoidable in \emph{any} realization.

Our realization of $pg$ makes the emergence quantitative. The glide
is the tile map $G(t_i,t_j) = (t_i{+}1,\,-t_j)$, a reflection
composed with a one-tile step; its square is the two-tile
translation, and the group's own translation lattice is
$\langle\Trow^2,\Tcol\rangle$, a proper index-2 sublattice of the
tile grid. The symmetry exists only at a scale coarser than any
single tile, and Lemma~\ref{lem:symmorphic} shows this coarsening is
forced, not chosen. A rigidity lemma (Lemma~\ref{lem:rigid})
sharpens the point from the payoff side: full translation invariance
of the cover payoffs forces them to be tile-periodic and hence blind
to all seam structure; any additional payoff symmetry comes from the
tile game itself, not from the spatial gluing, and the spatial
symmetry lives in the interaction pattern.

\subsection{Symmetry as Compression}

Two computational statements accompany the construction. They run
on separate tracks, one about the cover game and one about the cover
graph, and we keep them separate throughout: $M_{r,s}$ is not the
support complex of $\hat M_{r,s}$ (their vertex counts already
differ); Section~\ref{sec:poly} joins the tracks from the game side
with the polymatrix cover, and the canonical-complex refinement is
Open Problem~2.

First, a \emph{folding lemma} (Lemma~\ref{lem:fold}): the
translation-invariant cover game of the tile game $(A,B)$ has the
property that a profile is a Nash equilibrium if and only if its
tile-marginal is a Nash equilibrium of $(A,B)$. Consequently one Lemke-Howson run on the
tile, followed by lifting and the translation action, yields $rs$
distinct equilibria of the cover; conversely any cover equilibrium
projects to a tile equilibrium. Solving the cover and solving the
tile are linear-time equivalent, and detecting the tiling at the
payoff level is itself a linear-time scan
(Remark~\ref{rem:recognition}): the pipeline is detect, then fold.

Second, a \emph{fundamental domain theorem}
(Theorem~\ref{thm:fd}): quotienting the cover graph by its
translation group recovers the tile exactly, $\beq(M/\calT) =
\beq(K)$, for every straight-boundary cover, so the ambient search arena of
the quotient is exactly the tile's. For swap
boundaries the excess is exactly $\binom{m}{2}$, independent of
payoffs, of cover size, and of how many directions carry the swap.

Neither statement holds for a generic large game. All support
complexes in this paper are truncated at support size $k=2$. A
$10\times10$ cover of a $3\times3$ game has $\beq(M)=5401$; a
generic $30\times30$ game at the same truncation has
$\beq = 592{,}876$, larger by a factor of about $110$, and in general
the tiling suppresses cycle complexity by a factor of
$\Theta((rs)^{k-1})$. A generic game also has no payoff symmetry at
all outside a measure-zero set, while the cover carries a full
wallpaper group of seam symmetries by construction. The tiling
structure, not the size, is what matters
(Section~\ref{sec:notlarge}).

\subsection{Contributions}

\begin{enumerate}[label=(\arabic*)]
  \item \textbf{Realization theorem (Theorem~\ref{thm:main}).} All
    seventeen wallpaper groups are realized on plain straight
    covers, with every generator a machine-verified graph
    automorphism whose tile action is affine and whose local action
    is trivial; the infinite cover realizes $G$ exactly, and every
    finite realization is certified as the exact toroidal
    quotient. Types are
    identified by a crystallographic recognizer in exact rational
    arithmetic, validated on all 17 standard generator sets
    (Appendix~\ref{app:generators}).
  \item \textbf{The lattice classification
    (Lemma~\ref{lem:symmorphic}, Theorem~\ref{thm:lattice},
    Proposition~\ref{prop:template}).} A realization whose
    translation subgroup contains the full tile lattice exists
    exactly for the thirteen symmorphic groups; the four
    non-symmorphic groups are realized with translation lattices at
    index exactly two, the minimum Lemma~\ref{lem:symmorphic}
    permits.
    The earlier swap-glide template, whose tile action is a pure
    translation, generates only abelian groups and can present no
    non-symmorphic group; this corrects the glide-sector claims of a
    previous version.
  \item \textbf{Fundamental domain theorem (Theorem~\ref{thm:fd}).}
    Straight covers: $M/\calT\cong K$. Swap covers on $m\times m$
    tiles: $\beq(M/\calT) = \beq(K) + \binom{m}{2}$, regardless of
    the number of swap directions. Exact and payoff-independent,
    with complete proof and machine-verified data.
  \item \textbf{Folding, recognition, and CE stability
    (Lemmas~\ref{lem:rigid}, \ref{lem:fold},
    Theorem~\ref{thm:ce}).} Translation invariance forces
    tile-periodic payoffs; Nash computation on the cover game
    reduces to the tile and back after a linear-time detection scan;
    and the tiled correlated-equilibrium system of an $r$-fold cover
    has dimension exactly $r(d-q)+q$ (Theorem~\ref{thm:ce}), affine
    in $r$ with slope owned by the tile and no expansion regime. The
    polymatrix cover (Theorems~\ref{thm:polysym}
    and~\ref{thm:collapse}) then carries the symmetry outright:
    every toroidal wallpaper action is a group of genuine game
    automorphisms, glides included, equilibria collapse along any
    symmetry subgroup to a folded quotient equilibrium problem, and
    a decorated refinement achieves exactness: its automorphism
    group is exactly the toroidal wallpaper group
    (Theorem~\ref{thm:exact}).
\end{enumerate}

\subsection{Significance, Innovations, and Placement}
\label{sec:merits}

For a reader outside algorithmic game theory, the paper in one
paragraph. Take any two-player game, form the graph of its candidate
equilibrium supports, tile the plane with copies of that graph, and
choose rules for how neighboring copies communicate. We show that
every one of the seventeen crystallographic symmetry groups of the
plane acts on such tilings, by explicit closed-form generators that
a short script certifies exactly, and that the classical boundary
inside crystallography, symmorphic versus non-symmorphic, is
precisely the boundary between groups whose translations can be the
tile grid itself and groups that are forced to treat the tile as
half a step. Quotienting a tiling by its translations returns
exactly one tile, so the global object is no more complex than the
local one, and parallel local-global statements hold for the Nash equilibria
of the duplicated game and for a tiled correlated-equilibrium
system.

What is new. That some embedding of each group
exists is classical in hindsight, since the plain cover carries the
full square or hexagonal symmetry and every plane group sits inside
$p4m$ or $p6m$~\cite{conway2008}; we say so explicitly
(Remark~\ref{rem:aut}). The contributions are the parts with no
classical counterpart: closed-form generators realizing the standard
crystallographic action on a canonical equilibrium object, with
every automorphism, relation, and exact toroidal order machine
certified and independently confirmed by GAP's crystallographic
library; the lattice classification (Theorem~\ref{thm:lattice}),
an if-and-only-if with a three-line proof that turns the
symmorphic/non-symmorphic distinction into a statement about where
a group's translations may sit relative to the tile grid; the
fundamental domain theorem with its exact, payoff-independent
constant $\binom m2$, whose cleanness is a support-complex property
rather than a graph generality (Remark~\ref{rem:gameenters}); the
tile's square-filled product complex, whose first homology at $k=2$
is one copy of the exterior square of the standard representation
per player (Appendix~\ref{app:tile}); the
folding and detection pipeline, by which one tile solution yields a
full translation orbit of $rs$ equilibria of the cover game after a
linear-time scan; and the correlated-equilibrium dimension formula
$r(d-q)+q$, computed once per tile, with expansion provably
impossible.

The agenda. The graph statements and the game statements run on
clearly labeled tracks, and Section~\ref{sec:poly} joins them from
the game side: the polymatrix cover carries every toroidal wallpaper
action as a group of genuine game automorphisms, the glides
included, and equilibrium computation collapses along any symmetry
subgroup to a folded equilibrium problem on the quotient
(Theorems~\ref{thm:polysym} and~\ref{thm:collapse}), with one symmetric-equilibrium computation on the tile yielding a
wallpaper-invariant equilibrium of the $rs$-player game, the group
transporting every equilibrium to its orbit, and exactness achieved
by decoration: the
automorphism group of the decorated cover is exactly the toroidal
wallpaper group (Theorem~\ref{thm:exact}). The residual program is
stated precisely: a canonical pivot complex for the polymatrix
cover, and exactness on the graph side (Open Problems~1
and~2).

Where it sits. The paper is deliberately between communities: plane
crystallography~\cite{armstrong,conway2008}, topological graph
theory~\cite{grosstucker,frucht1939,sabidussi1960}, and equilibrium
computation~\cite{lh1964,daskalakis2009,chen2009,savani2006}. Each
ingredient is well developed inside its own field; the connection
among them, symmetry groups of the plane acting on equilibrium
search spaces with exact quotient and dimension consequences, does
not to our knowledge appear in any of the three literatures. Papers
whose contribution is a new connection between established areas,
with simple proofs and explicitly stated open programs, are the kind
this venue exists for.

\section{Preliminaries}
\label{sec:prelim}

\subsection{Support Complexes}

\begin{definition}[Support complex]
For $(A,B)\in\mathbb{R}^{m\times n}\times\mathbb{R}^{m\times n}$
and $k \geq 1$, the support complex $K(A,B)$ has node set
$V(K) = \{(S_i,S_j) : \emptyset\neq S_i\subseteq[m],
\emptyset\neq S_j\subseteq[n], |S_i|,|S_j|\leq k\}$,
with an edge between $(S_i,S_j)$ and $(S_i',S_j')$ iff they
differ by exactly one action added to or removed from one player's
support. The \emph{boundary nodes} $\partial K = \{(\{i\},\{j\})\}$
are the pure-strategy pairs, $|\partial K| = mn$.
\end{definition}

Let $H_k(m)$ be the \emph{support-edit graph}: vertices the
nonempty $S\subseteq[m]$ with $|S|\leq k$, edges the pairs with
$|S\,\triangle\,S'| = 1$, so that
$T_k(m) := \sum_{s=1}^k \binom{m}{s} = |V(H_k(m))|$ and
$D_k(m) := \sum_{s=1}^{k-1}\binom{m}{s}(m-s) = |E(H_k(m))|$. Then
$K(A,B) = H_k(m)\,\square\,H_k(n)$ as a Cartesian product, connected
for $m,n\geq2$ and $k\geq2$, with
\[
  \beq(K) = D_k(m)T_k(n) + D_k(n)T_k(m) - T_k(m)T_k(n) + 1
\]
(Theorem~\ref{thm:product}); the proof, a saturation bound
($\beq(K)\geq|V(K)|$ for $m,n\geq3$), the
$S_m\times S_n$-equivariant homology of the square-filled complex,
and a genericity theorem are collected in
Appendix~\ref{app:tile}.

All computations in this paper use $k=2$; the theory is stated for
general $k$ where this costs nothing. $K(A,B)$, as defined, does not
depend on payoffs; payoff classes enter only through the
labeled-cover refinement of Remark~\ref{rem:labeled}, and for
generic payoffs the equilibrium-feasible support pairs form a
discrete independent set inside $K$
(Theorem~\ref{thm:genericity}), so $\beq(K)$ counts the cycles of
the ambient combinatorial arena, not equilibrium topology. Each
Lemke-Howson pivot changes the support pair by at most one addition
and one deletion, so pivot orbits project to walks of step length at
most two in $K$; we use this only as motivation for the edge
relation and claim no bounds from it.

\paragraph{Game types.} Table~\ref{tab:results} refers to payoff
classes of the tile: \emph{generic} (i.i.d.\ payoffs);
\emph{ctrsy.}\ (centrosymmetric, $A_{ij}=A_{m-1-i,\,n-1-j}$ and
likewise $B$); \emph{sym.}, \emph{periodic}, $\bZ_k$-\emph{cyc.},
and $\bZ_k$-\emph{mir.}\ denote payoff classes invariant under,
respectively, single-axis index reversal, index translation, cyclic
index rotation, and the dihedral closure of the cyclic class.
Generators whose local actions are trivial are payoff symmetries of
every tile, so the corresponding rows carry no payoff requirement.

\subsection{Wallpaper Groups and the Symmorphic Line}

The wallpaper groups are the 17 symmetry groups of periodic planar
tilings, the complete classification of discrete cocompact groups of
Euclidean plane isometries containing two independent
translations~\cite{armstrong,conway2008}. Thirteen are
\emph{symmorphic}: an origin can be chosen so that every element is
a point-group isometry followed by a lattice translation
($p1$, $p2$, $pm$, $cm$, $pmm$, $cmm$, $p3$, $p3m1$, $p31m$, $p4$,
$p4m$, $p6$, $p6m$). Four are \emph{non-symmorphic}: no such origin
exists; equivalently, the extension of the point group by the
translation lattice does not split ($pg$, $pmg$, $pgg$, $p4g$).

This classical line is exactly the realization boundary of our
framework. The thirteen symmorphic groups are realized on straight
covers with affine tile actions and trivial local actions
(Theorem~\ref{thm:main}); symmorphic means precisely that an
integral choice of translation parts exists, and
Lemma~\ref{lem:symmorphic} shows the converse: an integral-affine
group whose translations contain the full tile lattice is
symmorphic. A non-symmorphic group must therefore act with its own
translation lattice a proper sublattice of the tile grid; all four
non-symmorphic groups are realized this way
(Theorem~\ref{thm:main}(ii)), and Theorem~\ref{thm:lattice} states
the resulting equivalence. Swap boundaries, which exchange player roles
across a seam, are not needed for these realizations; they remain
central to the fundamental domain theorem (swap case of
Theorem~\ref{thm:fd}) and to the symmetry-breaking program of Open
Problem~1, and they appear in the legacy configurations retained in
Tables~\ref{tab:results} and~\ref{tab:pres} for their cover data.

\subsection{Presentations and Toroidal Quotients}

A \emph{group presentation} $G = \langle x_1,\ldots,x_k \mid
w_1=1,\ldots,w_\ell=1\rangle$ defines $G$ as the free group on
$\{x_i\}$ modulo the normal closure of the words $\{w_j\}$.
Verifying a presentation concretely means exhibiting permutations
$\pi_i$ of a finite set such that each word
$w_j(\pi_1,\ldots,\pi_k)$ evaluates to the identity permutation.
Relations alone certify only a quotient. On a finite cover the
group acts through the torus $\bZ^2/(r\bZ\oplus s\bZ)$, and the
\emph{toroidal quotient} is the image of $G$ in its affine action on
that torus, a finite group of order
$|\Lambda_G/(\Lambda_G\cap(r\bZ\oplus s\bZ))|\cdot|\Phi|$, where
$\Lambda_G$ is the translation lattice of $G$ and $\Phi$ its point
group (this is $rs\,|\Phi|$ when the translations are the full tile
lattice), provided the period lattice is point-group invariant and
no element with nontrivial point part acts trivially on the torus; a
half-turn on the $2\times2$ torus, where $-I\equiv I$, shows the
proviso is needed. Every cover size used here satisfies it, as the
computed orders certify. We report the order of the generated permutation group for
every realization; equality with this quantity certifies the exact
toroidal quotient. The full
infinite group $G$ is realized on the infinite cover
(Theorem~\ref{thm:main}); this is the standard relationship between
crystallographic groups and their actions on tori.

\section{The Multigame Cover}
\label{sec:cover}

\subsection{Construction}

A map $\varphi: V(M) \to V(M)$ is a \emph{graph automorphism}
of a graph $M$ if $(u,w)\in E(M)$ if and only if
$(\varphi(u),\varphi(w))\in E(M)$.

\begin{definition}[Multigame cover $M_{r,s}$]\label{def:cover}
The $(r\times s)$-cover of $K = K(A,B)$ is the graph with
\[
V(M_{r,s}) = \{0,\ldots,r{-}1\}\times\{0,\ldots,s{-}1\}\times V(K),
\]
$rs\,|V(K)|$ nodes, all \emph{disjoint}. Edges are of four types:
\begin{itemize}
  \item \textbf{Intra-tile:} $(t_i,t_j,u)\sim(t_i,t_j,v)$ iff $u\sim v$ in $K$.
  \item \textbf{Inter-tile col:} $(t_i,t_j,v)\sim(t_i,(t_j{+}1)\bmod s, v')$
    for $v\in\partial K$; $v'=v$ (straight) or, when $m=n$,
    $v'=\swp(v)=(\{j\},\{i\})$ (swap).
  \item \textbf{Inter-tile row:} analogously in the row direction.
  \item \textbf{Oblique (hex):} $(t_i,t_j,v)\sim((t_i{+}1)\bmod r,(t_j{-}1)\bmod s,v)$
    for $v\in\partial K$; used for triangular and hexagonal groups.
\end{itemize}
Configuration flags: \texttt{row\_swap}, \texttt{col\_swap} (which
boundary directions swap player roles), \texttt{uniform\_swap}
(every boundary swaps, versus alternating), \texttt{hex\_lattice}
(add oblique edges). The \emph{infinite cover} $M_\infty$ is defined
by the same rules with tiles indexed by $\bZ\times\bZ$ and no
modular reduction. Quotient graphs are taken in the category of
simple graphs: vertices are orbits, loops are discarded, and
parallel edges are merged. A cover is \emph{straight} when every
seam, in every direction present (including the oblique hex edges),
uses the identity rule $v'=v$; the hex flag adds a lattice direction
and does not change the local rule.
\end{definition}

\begin{figure}[t]
\centering
\begin{tikzpicture}[scale=0.85,
  tile/.style={draw=gray!60,fill=blue!8,rounded corners=2pt,minimum width=1.2cm,minimum height=1.2cm},
  inter/.style={draw=red!70,thick,dashed},
  ginter/.style={draw=green!60!black,thick,densely dotted},
]
\foreach \ti in {0,1,2} {
  \foreach \tj in {0,1,2} {
    \node[tile] at (1.5*\tj, -1.5*\ti) {};
  }
}
\foreach \ti in {0,1,2} {
  \foreach \tj in {0,1,2} {
    \node[font=\footnotesize] at (1.5*\tj, -1.5*\ti) {$K$};
  }
}
\draw[inter] (0.6,0.1) -- (0.9,0.1) node[midway,above,font=\tiny,red!70] {straight};
\draw[inter] (2.1,0.1) -- (2.4,0.1);
\draw[ginter] (0.1,-0.6) -- (0.1,-0.9) node[midway,right,font=\tiny,green!60!black] {swap};
\draw[ginter] (0.1,-2.1) -- (0.1,-2.4);
\draw[-{Stealth},blue!80,thick] (4.0,0) -- (4.0,-1.5) node[midway,right,font=\small] {$\Trow$};
\draw[-{Stealth},blue!80,thick] (0,0.8) -- (1.5,0.8) node[midway,above,font=\small] {$\Tcol$};
\draw[decorate,decoration={brace,amplitude=5pt,mirror}] (-0.7,0.65) -- (-0.7,-3.65) node[midway,left,font=\small,xshift=-3pt] {$r$};
\draw[decorate,decoration={brace,amplitude=5pt}] (-0.65,0.65) -- (3.65,0.65) node[midway,above,font=\small,yshift=3pt] {$s$};
\end{tikzpicture}
\caption{The $r\times s$ multigame cover $M_{r,s}$. Each tile is a copy of
$K(A,B)$ with nodes disjoint from all other tiles. Dashed red edges:
straight inter-tile connections. Dotted green: swap connections where
player roles are exchanged. Translations $\Trow$ and $\Tcol$ shift one
tile in each direction; they are free permutations because no nodes
are shared.}
\label{fig:cover}
\end{figure}
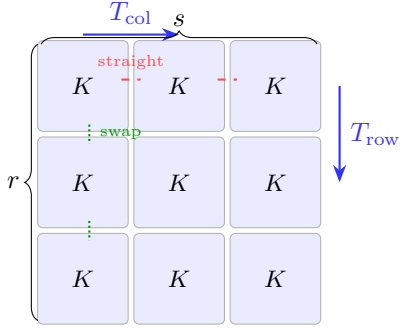

\subsection{The Cover Game and Translation Rigidity}
\label{sec:covergame}

The cover graph encodes a search space; the equilibrium statements
of Section~\ref{sec:fold} require a game.

\begin{definition}[Cover game]\label{def:covergame}
The cover game $\hat M_{r,s}$ of $(A,B)$ has row strategy set
$\bZ_r\times[m]$ and column strategy set $\bZ_s\times[n]$, with
payoffs
$\hat A_{(a,i),(b,j)} = A_{ij}$ and $\hat B_{(a,i),(b,j)} = B_{ij}$
for all slots $a\in\bZ_r$, $b\in\bZ_s$.
\end{definition}

The translations act on strategies by slot shifts:
$\Trow:(a,i)\mapsto(a{+}1,i)$ on row strategies (identity on column
strategies), and symmetrically $\Tcol$. This payoff structure is not
a choice but a consequence:

\begin{lemma}[Translation rigidity]\label{lem:rigid}
A payoff pair on $(\bZ_r\times[m])\times(\bZ_s\times[n])$ is
invariant under both $\Trow$ and $\Tcol$ (as game symmetries) if and
only if it has the form of Definition~\ref{def:covergame}.
\end{lemma}

\begin{proof}
Invariance under $\Trow$ reads
$\hat A_{(a+1,i),(b,j)} = \hat A_{(a,i),(b,j)}$ for all arguments,
so $\hat A$ is independent of $a$; invariance under $\Tcol$ gives
independence of $b$. Likewise for $\hat B$. The converse is
immediate.
\end{proof}

\begin{remark}[Symmetry of payoffs versus symmetry of seams]
\label{rem:labeled}
Rigidity says that exact wallpaper symmetry at the payoff level
forces the payoffs to be blind to the seams: $\hat M_{r,s}$ cannot
distinguish a straight boundary from a swap boundary. The seam
structure is carried by the cover graph $M_{r,s}$, not by the payoff
matrix. When a generator's local action is nontrivial and the tile
payoffs lie in the matching class of Table~\ref{tab:results}, the
generator additionally preserves payoffs and the realized group acts
on the payoff-labeled cover; generators with trivial local action
(all of Theorem~\ref{thm:main}) act on the labeled cover for every
payoff class.
\end{remark}

\subsection{Why Nodes Must Be Disjoint}

Disjointness is not a technicality; it is essential. If boundary
nodes were shared between adjacent tiles (as in a standard torus
identification), then the translation $\Trow(t_i,t_j,v) =
((t_i{+}1)\bmod r, t_j, v)$ would fix shared boundary nodes, making
it non-free. Free translations are required for the fundamental
domain theorem: a non-free $\Trow$ gives $|V(K)| < |V(M)|/rs$,
breaking the orbit bijection $V(M/\calT)\leftrightarrow V(K)$.
With disjoint nodes, $\Trow$ is free: $(t_i{+}1)\bmod r = t_i$ has
no solution for $r\geq 2$.

\subsection{The Semidirect Product Requirement}

Every nontrivial wallpaper group has the structure
$G = (\bZ\times\bZ)\rtimes\Phi$ (symmorphic) or a non-split
extension of $\Phi$ by $\bZ\times\bZ$ (non-symmorphic); either way,
the point-group part acts nontrivially on the translations,
witnessed by relations like $R\,\Trow\,R^{-1} = \Tcol$ for $p4$. To
obtain such relations, the generator $R$ must act on \emph{tile
coordinates} $(t_i,t_j)$ with a nontrivial linear part. If $R$
acts only on the local node $v$, then $R\,\Trow\,R^{-1} = \Trow$
always (a direct product), and no nontrivial wallpaper group can be
realized. Proposition~\ref{prop:template} pushes this one step
further: if the tile action is a translation and only the local
action is nontrivial, the generated group is abelian, so no glide
generator of that shape can help either. The linear part on tile
coordinates is where all crystallographic content lives.

\subsection{A Cover Is Not a Generic Large Game}
\label{sec:notlarge}

A multigame cover of $K(A,B)$ occupies the footprint of an
$rm\times sn$ game. It is not one in disguise, for three reasons.

First, \emph{cycle complexity}. At truncation $k$, a generic
$p\times q$ game has $\beq \sim \Theta(p^k q^k)$, so a generic
$rm\times sn$ game has $\beq\sim\Theta((rs)^k(mn)^k)$, while the
cover has $\beq(M) \sim rs\cdot\beq(K)$; the tiling suppresses cycle
complexity by a factor of order $(rs)^{k-1}$. Concretely at $k=2$: a
$10\times10$ cover of a $3\times3$ game has
$\beq(M) = 100\cdot 37 + 2\cdot100\cdot 9 - 99 = 5401$, while a
generic $30\times30$ game has $\beq = 592{,}876$, about $110$ times
larger.

Second, \emph{symmetry}. For a random $rm\times sn$ game, the payoff
symmetry group is trivial with probability 1 (the conditions for any
nontrivial symmetry form a measure-zero set in payoff space). The
cover carries a full wallpaper group of seam symmetries by
construction, a structure no generic large game possesses.

Third, \emph{the fundamental domain theorem}. For a generic large
game there is no reason for the quotient complexity to equal the
tile complexity; the equality $\beq(M/\calT) = \beq(K)$ is a theorem
about the tiling structure with no analogue for a generic large
game; it is the graph track's exact counterpart of the game track's
amortization. And for generic payoffs the equilibrium-feasible
support pairs are a discrete independent set in the tile complex
(Theorem~\ref{thm:genericity}), which is exactly why every $\beq$
statement here is about the ambient arena and its symmetries.

\section{Main Results}
\label{sec:main}

\subsection{The Realization Theorem}

\begin{lemma}[Straight-cover automorphisms]\label{lem:straightauto}
Let $L$ be the direction set of the cover, $\{\pm e_1,\pm e_2\}$ for
the square lattice or $\{\pm e_1,\pm e_2,\pm(e_1{-}e_2)\}$ for the
hexagonal one. If an affine tile map $t\mapsto Nt+\mu$ satisfies
$N(L)=L$, then $(t,v)\mapsto(Nt+\mu,\,v)$ is a graph automorphism of
the infinite plain straight cover, and of every finite one on which
the map descends, that is, with $N(r\bZ\oplus s\bZ) =
r\bZ\oplus s\bZ$; this holds automatically when $r=s$, which every
cover in this paper satisfies.
\end{lemma}

\begin{proof}
Intra-tile edges are preserved because the local action is the
identity. A straight seam edge joins $(t,v)$ to $(t+\delta,v)$ for
some $\delta\in L$ and $v\in\partial K$; its image joins
$(Nt{+}\mu,v)$ to $(Nt{+}\mu{+}N\delta,v)$, which is a seam edge of
the cover since $N\delta\in L$ and every direction of $L$ carries the
identity rule. The descent condition makes the map well defined on
the torus; it is a bijection there since $N\in GL_2(\bZ)$, and the
same argument applies to its inverse.
\end{proof}

Each generator of this paper has a linear part visibly permuting the
relevant direction set, so Lemma~\ref{lem:straightauto} proves the
automorphism claims analytically; the machine checks certify these
facts together with the relations and the closure orders, rather
than substituting for proof.

Throughout, a \emph{realization in this framework} is a group of
permutations of $V(M_{r,s})$ of the form
$(t,v)\mapsto(\varphi(t),\lambda(v))$, with $\varphi$ an integral
affine map of the tile coordinates and $\lambda$ a permutation of
$V(K)$, each generator a graph automorphism of the cover.

\begin{theorem}[Realization]\label{thm:main}
\hspace{0pt}
\begin{enumerate}[label=(\roman*)]
  \item \textbf{(Symmorphic groups.)} Each of the thirteen
    symmorphic wallpaper groups is realized on a plain straight
    cover (Table~\ref{tab:results}; generators in
    Appendix~\ref{app:generators}): every generator is a verified
    graph automorphism with trivial local action, its tile action is
    the standard affine crystallographic action of $G$ on $\bZ^2$,
    the relations of Table~\ref{tab:pres} hold as permutation
    identities, and on the infinite cover $M_\infty$ the generated
    group is isomorphic to $G$. The two hexagonal mirror families, which share the rotation and
    translations and differ only in the mirror, are assigned their
    names by the centers-on-mirrors certificate of
    Appendix~\ref{app:generators}. Every generated permutation group
    on its finite cover has order exactly $rs\,|\Phi|$: the toroidal
    quotient is realized exactly.
  \item \textbf{(Non-symmorphic groups.)} Each of $pg$, $pmg$,
    $pgg$, $p4g$ is realized on the plain straight $4\times4$ cover
    by glides whose tile actions reflect:
    $pg = \langle G, \Tcol\rangle$,
    $pmg = \langle G, S, \Tcol\rangle$,
    $pgg = \langle g, h\rangle$,
    $p4g = \langle R_{4g}, g\rangle$, where, with tile
    arithmetic mod $(r,s)$ and trivial action on local nodes,
    \begin{gather*}
      G(t_i,t_j) = (t_i{+}1,\,-t_j),\qquad
      S(t_i,t_j) = (-t_i,\,t_j),\qquad
      R_{4g}(t_i,t_j) = (-t_j,\,t_i),\\
      g(t_i,t_j) = (t_j{+}1,\,t_i),\qquad
      h(t_i,t_j) = (1{-}t_j,\,-t_i).
    \end{gather*}
    All generators are
    verified graph automorphisms; the relations of
    Table~\ref{tab:pres} hold; the generated groups are the exact
    nonabelian toroidal quotients (orders $16$, $32$, $32$, $64$);
    and the infinite-cover types are certified by the recognition
    invariants of Appendix~\ref{app:generators}. The translation
    lattices are $\langle\Trow^2,\Tcol\rangle$ for $pg$ and $pmg$
    and the diagonal lattice
    $\langle\Trow\Tcol,\,\Trow\Tcol^{-1}\rangle$ for $pgg$ and
    $p4g$, each of index exactly $2$, the minimum
    Lemma~\ref{lem:symmorphic} permits; indeed $pgg = \langle g,\,
    R_{4g}^{2}g\rangle$ sits inside the $p4g$ realization as an
    index-$2$ subgroup.
\end{enumerate}
All verifications are payoff-independent and reproduced by the
deterministic script \path{verify_wallpaper_fixes.py}, included
with this submission as an ancillary file together with the
independent audits \path{verify_new_results.py} and
\path{verify_exactness.py}, the recomputation tool
\path{recompute_ce_dq.py}, and the GAP cross-check; the legacy five-seed suite additionally
reports $365/365$ relation instances for the configurations of an
earlier version.
\end{theorem}

\begin{proof}[Proof sketch]
The automorphism claims are Lemma~\ref{lem:straightauto}: each
generator's linear part permutes the relevant direction set, and the
lemma covers the finite and infinite covers alike. The relation
checks are finite computations. The linear
parts of the tile actions are the standard point-group matrices on
the square or hexagonal lattice (for example,
$R_3:(t_i,t_j)\mapsto(-t_i{-}t_j,\,t_i)$ has matrix
$\left(\begin{smallmatrix}-1&-1\\ 1&0\end{smallmatrix}\right)$ of
order 3), and the finite formulas differ from the zero-offset standard forms
by composition with available lattice translations, so replacing
each occurrence of $r{-}1{-}t$ or $(-t)\bmod r$ by $-t$ over $\bZ$
changes no generated group and puts the generators in standard
affine form, in which the relations hold as identities of affine
maps. Since all local actions are trivial,
the resulting homomorphism factors through the tile action, and
injectivity on $\bZ^2$ suffices: every nonidentity plane isometry in
$G$ acts nontrivially on the tile lattice, because a nonidentity
linear part moves some lattice vector and a pure translation moves
every tile. The exact-order claims are closure computations, one per group,
each matching $|\Lambda_G \bmod (r,s)|\cdot|\Phi|$ (for instance
$54$ for the hexagonal mirror groups, $128$ for $p4m$, $432$ for
$p6m$, and $16$, $32$, $16$, $64$ for $pg$, $pmg$, $pgg$, $p4g$). For the non-symmorphic realizations the infinite-cover type
is identified by the recognition invariants of
Appendix~\ref{app:generators} (maximal rotation order, mirror
classes, lattice splitting index, centers-on-mirrors), computed in
exact rational arithmetic and validated against standard generator
sets for all 17 types.
\end{proof}

\subsection{The Symmorphic Boundary}

\begin{lemma}[Symmorphic boundary]\label{lem:symmorphic}
Let $\Gamma \subseteq \mathrm{Aff}(\bZ^2)$ be a wallpaper group
acting by integral affine maps on the tile lattice. If the
translation subgroup of $\Gamma$ contains the full tile lattice
$\bZ^2$ (that is, every unit tile translation lies in $\Gamma$),
then $\Gamma$ is symmorphic. Consequently, any
integral-affine realization of a non-symmorphic group has its
translation lattice at index at least $2$ in the tile grid.
\end{lemma}

\begin{proof}
Let $(N,\mu)\in\Gamma$ with linear part $N$ and integral offset
$\mu$. Since $(I,\lambda)\in\Gamma$ for every $\lambda\in\bZ^2$, the
composite $(N,\mu)\circ(I,-N^{-1}\mu) = (N,0)$ lies in $\Gamma$.
Thus every point-group element lifts to $\Gamma$ with zero offset:
the point group lifts at the origin, which is the definition of
symmorphic.
\end{proof}

\begin{theorem}[Lattice classification]\label{thm:lattice}
A wallpaper group $G$ admits a realization in this framework whose
translation subgroup contains the full tile lattice if and only if
$G$ is symmorphic.
\end{theorem}

\begin{proof}
The thirteen symmorphic realizations of Theorem~\ref{thm:main}(i)
use the full tile lattice. Conversely, the tile-action homomorphism
of any realization is faithful: its kernel consists of elements
acting only on local nodes, hence of finite order, and a wallpaper
group has no nontrivial finite normal subgroup (a nonidentity
isometry of finite order has a fixed point, and conjugating by the
infinitely many translations produces infinitely many distinct
elements). The tile actions are therefore an integral-affine copy of
$G$, and Lemma~\ref{lem:symmorphic} applies.
\end{proof}

\begin{proposition}[Swap-glide templates are abelian]
\label{prop:template}
Let $a, b$ be the tile translations and let $G$ be any permutation
of $V(M)$ of the form $G(t,v) = (t + \delta,\, \sigma(v))$ with
$\delta$ a fixed tile vector and $\sigma$ a fixed involution of
$V(K)$. Then $\langle a, b, G\rangle$ is abelian; more generally,
$G$ centralizes every tile translation, whereas a faithful image of
an orientation-reversing crystallographic element cannot, so no
generator of this form can represent a glide or a reflection in any
realization whose translations act faithfully on the tile lattice,
whether or not other generators are present.
\end{proposition}

\begin{proof}
The local factor $\sigma$ acts identically in every tile and
commutes with all tile translations; the tile parts of $a$, $b$, $G$
are translations of $\bZ_r\times\bZ_s$ and commute among themselves,
so all three generators commute pairwise. For the general claim:
$G T_{\lambda} G^{-1} = T_{\lambda}$ for every tile translation
$T_{\lambda}$, while an orientation-reversing element $g$ of a
wallpaper group satisfies $g\,t\,g^{-1} = $ the translation by
$N\lambda$ with reflection part $N\neq I$; if the realization's
translations act faithfully, these conjugation behaviors are
incompatible.
\end{proof}

\paragraph{From template failure to realization.} An earlier
version of this paper realized the glide sector through
swap-boundary configurations whose glide generators had
pure-translation tile actions. Proposition~\ref{prop:template}
shows their glide generators centralize all translations and so
cannot represent glides (for the $pg$ configuration the whole
generated group is abelian of order $32$, machine-confirmed, while
$pg$ is the nonabelian Klein bottle group). The corrected mechanism, reflecting
tile actions over a proper sublattice, realizes all four
non-symmorphic groups (Theorem~\ref{thm:main}(ii)); in particular
the $p4g$ obstruction reported earlier was an artifact of the
template normalization (Appendix~\ref{app:doubled}). The legacy
configurations are retained in Tables~\ref{tab:results}
and~\ref{tab:pres} (marked $\dagger$) for their cover and quotient
data, which are unaffected.

\begin{table}[t]
\centering
\caption{Configurations and verification data. $\beq(K)$: tile Betti
number. $\beq(M/\calT)$: quotient Betti number. Thm~\ref{thm:fd}:
quotient matches the prediction ($\beq(K)$ for straight covers;
$\beq(K)+\binom m2$ for swap covers). Rels: verified/total listed
relations. The legacy swap-template configurations of an earlier
version are quarantined in Appendix~\ref{app:doubled}
(Proposition~\ref{prop:template}). For $r,s\geq 3$,
$\beq(M) = rs\,\beq(K) + 2rs\,mn - (rs{-}1)$ on the square lattice
(e.g.\ $p1$: $9\cdot37+162-8=487$), plus $rs\,mn$ for hex covers
(e.g.\ $p3$: $487+81=568$); swap boundaries do not change the count.}
\label{tab:results}
\small
\setlength{\tabcolsep}{4pt}
\begin{tabular}{llrrrrcc}
\toprule
Group & Cover & $|V(K)|$ & $\beq(K)$ & $\beq(M)$ & $\beq(M/\calT)$ & Thm~\ref{thm:fd} & Rels \\
\midrule
$p1$   & $3\times3$    &  36 &  37 &   487 &  37 & \checkmark & 1/1 \\
$p2$   & $3\times3$    &  36 &  37 &   487 &  37 & \checkmark & 4/4 \\
$pm$   & $3\times3$    &  36 &  37 &   487 &  37 & \checkmark & 4/4 \\
$cm$   & $4\times4$    &  36 &  37 &   865 &  37 & \checkmark & 3/3 \\
$pmm$  & $3\times3$    & 100 & 141 &  1549 & 141 & \checkmark & 8/8 \\
$cmm$  & $4\times4$    &  36 &  37 &   865 &  37 & \checkmark & 7/7 \\
$p4$   & $4\times4$    & 100 & 141 &  2753 & 141 & \checkmark & 4/4 \\
$p4m$  & $4\times4$    & 100 & 141 &  2753 & 141 & \checkmark & 7/7 \\
$p3$   & $3\times3$,h  &  36 &  37 &   568 &  37 & \checkmark & 3/3 \\
$p3m1$ & $3\times3$,h  &  36 &  37 &   568 &  37 & \checkmark & 6/6 \\
$p31m$ & $3\times3$,h  &  36 &  37 &   568 &  37 & \checkmark & 6/6 \\
$p6$   & $6\times6$,h  & 441 & 820 & 33373 & 820 & \checkmark & 4/4 \\
$p6m$  & $6\times6$,h  & 441 & 820 & 33373 & 820 & \checkmark & 7/7 \\
$pg$   & $4\times4$    &  36 &  37 &   865 &  37 & \checkmark & 3/3 \\
$pmg$  & $4\times4$    &  36 &  37 &   865 &  37 & \checkmark & 6/6 \\
$pgg$  & $4\times4$    &  36 &  37 &   865 &  37 & \checkmark & 3/3 \\
$p4g$  & $4\times4$    &  36 &  37 &   865 &  37 & \checkmark & 4/4 \\
\bottomrule
\multicolumn{8}{l}{\small h = hex lattice. All rows are realizations
(Theorem~\ref{thm:main}); payoffs enter nowhere.}
\end{tabular}
\end{table}

\begin{table}[t]
\centering
\caption{Relation sets verified as permutation identities. $a$, $b$
denote the group's lattice translations; for the realizations of
Theorem~\ref{thm:main}(i) these are $\Trow$, $\Tcol$ (or swapped;
Appendix~\ref{app:labeling}); for $pg$ and $pmg$ they are $\Trow^2$
and $\Tcol$, and for $pgg$ and $p4g$ the diagonal translations $d$,
$d'$, with the $p4g$ rotation in the $R_{4g}$ orientation of
Appendix~\ref{app:generators}.}
\label{tab:pres}
\small
\begin{tabular}{ll}
\toprule
Group & Relations beyond $ab=ba$ \\
\midrule
$p1$   & (none) \\
$p2$   & $R^2=1,\; RaR=a^{-1},\; RbR=b^{-1}$ \\
$pm$   & $S^2=1,\; SaS=a,\; SbS=b^{-1}$ \\
$cm$   & $S^2=1,\; SaS=b$ \\
$pmm$  & $R^2=S^2=1,\; RaR=a^{-1},\; RbR=b,\; SaS=a,\; SbS=b^{-1},\; RS=SR$ \\
$cmm$  & $R^2=S^2=1,\; RaR=a^{-1},\; RbR=b^{-1},\; SaS=b,\; (SR)^2=1$ \\
$p4$   & $R^4=1,\; RaR^{-1}=b,\; RbR^{-1}=a^{-1}$ \\
$p4m$  & $R^4=S^2=1,\; RaR^{-1}=b,\; RbR^{-1}=a^{-1},\; SaS=b,\; SRS=R^{-1}$ \\
$p3$   & $R^3=1,\; RaR^{-1}=a^{-1}b$ \\
$p3m1$ & $R^3=S^2=1,\; RaR^{-1}=a^{-1}b,\; SbS=b^{-1},\; SRS=R^{-1}$ \\
$p31m$ & $R^3=S^2=1,\; RaR^{-1}=a^{-1}b,\; SaS=b^{-1},\; SRS=R^{-1}$ \\
$p6$   & $R^6=1,\; RaR^{-1}=b,\; RbR^{-1}=a^{-1}b$ \\
$p6m$  & $R^6=S^2=1,\; RaR^{-1}=b,\; RbR^{-1}=a^{-1}b,\; SaS=b^{-1},\; SRS=R^{-1}$ \\
$pg$   & $G^2=a,\; GbG^{-1}=b^{-1}$ \\
$pmg$  & $G^2=a,\; GbG^{-1}=b^{-1},\; S^2=1,\; SGS=G^{-1},\; SbS=b,\; SaS=a^{-1}$ \\
$pgg$  & $g^2=d,\; h^2=d',\; (gh)^2=1$\quad(diagonal lattice) \\
$p4g$  & $R^4=1,\; g^2=d,\; gRg^{-1}=dR^{-1}$\quad($R=R_{4g}$; $d,d'{=}RdR^{-1}$ diagonal; $dd'{=}d'd$) \\
\bottomrule
\end{tabular}
\end{table}

\begin{remark}[$G\subseteq\Aut(M)$, not $\Aut(M)=G$]
\label{rem:aut}
Theorem~\ref{thm:main} is an embedding statement, and it is worth
stating plainly which part is classical. The plain straight cover
carries the ambient $p4m$ (square) or $p6m$ (hexagonal) action on
tile coordinates, and every plane group is a subgroup of finite
index of $p4m$ or of $p6m$~\cite{conway2008}, so bare existence of
some embedding is classical in hindsight; likewise, Frucht's
theorem~\cite{frucht1939}, extended to infinite groups by de Groot
and Sabidussi~\cite{sabidussi1960}, realizes any abstract group as
the full automorphism group of some ad hoc graph. The content here
is different in kind, because Frucht-type results realize the
abstract group while Theorem~\ref{thm:main} realizes the
\emph{action}: a Frucht graph for $pg$ contains an element
satisfying the glide relations but no geometry in which it glides,
so the lattice classification of Theorem~\ref{thm:lattice}, the
forced sublattice indices, the fundamental domain theorem, and the
emergence statements of Section~\ref{sec:interp} do not even
typecheck there, and the graph itself carries no complexity
semantics. What attaches here, to closed-form generators acting on a
canonical equilibrium object by the standard crystallographic
action, is machine certification with exact toroidal orders and the
complexity consequences of Theorem~\ref{thm:fd} and
Corollary~\ref{cor:amortized}. Conversely, Frucht-type
constructions achieve $\Aut(X)=G$ exactly, which we do not. $\Aut(M)$ itself is strictly larger
than $G$: the graph $K(A,B)$ does not depend on payoffs, so uniform
strategy relabelings applied identically in every tile act as a
local gauge group, and the ambient lattice symmetry is present as
well. Cutting the decorated symmetry group down to exactly $G$ times
gauge, using swap and richer seam alphabets as symmetry breakers, is
Open Problem~1.
\end{remark}

\subsection{Fundamental Domain Theorem}

\begin{theorem}[Fundamental domain]\label{thm:fd}
Let $M = M_{r,s}$ ($r,s\geq 2$) be a cover of a connected tile
complex $K$, let $\calT = \langle\Trow,\Tcol\rangle \cong
\bZ_r\times\bZ_s$, and let $M/\calT$ be the simple quotient graph.
Then $|V(M/\calT)| = |V(K)|$, and:
\begin{enumerate}
  \item \textbf{Straight cover} (all boundaries straight, including
    any oblique hex boundaries): $M/\calT \cong K$; in particular
    $\beq(M/\calT) = \beq(K)$.
  \item \textbf{Swap cover} ($m=n$; a nonempty set of boundary
    directions carries the uniform player swap, all others
    straight):
    \[
      \beq(M/\calT) \;=\; \beq(K) + \binom{m}{2},
    \]
    independent of payoffs, of $(r,s)$, and of how many directions
    carry the swap.
\end{enumerate}
\end{theorem}

\begin{proof}
$\calT$ acts freely on $V(M)$ (nodes are disjoint;
Section~\ref{sec:cover}), and the orbit of $(t_i,t_j,v)$ is
$\{(t_i',t_j',v)\}$, so orbits biject with $V(K)$:
$|V(M/\calT)| = |V(K)|$.

Intra-tile edges descend to the edges of $K$: the edge
$\{(t,u),(t,v)\}$ maps to $\{[u],[v]\}$, distinct orbits since
$u\neq v$. Every straight seam edge, in the row, column, or oblique
direction, joins $(t,v)$ to a $\calT$-translate of $(t,v)$: both
endpoints lie in the orbit of $v$, so the edge becomes a loop and is
discarded. This proves (i): $M/\calT$ has exactly the vertices and
edges of $K$.

For (ii), a swap seam edge joins $(t,v)$ to a translate of
$(t,\swp(v))$, descending to the pair $\{[v],[\swp(v)]\}$. For
$v=(\{i\},\{j\})$ this is a loop iff $i=j$; the off-diagonal
boundary nodes are partitioned into the pairs
$\{v,\swp(v)\}$, so the set of new quotient edges is exactly
\[
  \bigl\{\{v,\swp(v)\} : v\in\partial K,\ \swp(v)\neq v\bigr\},
  \qquad \text{of size } \binom{m}{2},
\]
and this set is the same whichever boundary directions carry the
swap, so multiple swap directions contribute the identical edge set,
which the simple quotient merges. None of these edges is already
present in $K$: $v$ and $\swp(v)$ differ in \emph{both} supports
simultaneously, while edges of $K$ change exactly one support.
No other identifications occur: all $rs$ translates of an edge of
$K$ descend to the same quotient edge, distinct edges of $K$ to
distinct quotient edges, and the swap pairs are pairwise disjoint
(each off-diagonal $v$ lies in exactly one pair) and disjoint from
$E(K)$. The quotient edge count is therefore exactly
$|E(K)| + \binom m2$, and the quotient contains $K$ as a spanning
subgraph, hence is connected, so
$\beq(M/\calT) = |E(K)| + \binom m2 - |V(K)| + 1 = \beq(K) +
\binom m2$.
\end{proof}

The quotient throughout is the reduced simple quotient of
Definition~\ref{def:cover}; the topological multigraph quotient
additionally retains one loop per seam orbit, and its cycle rank
exceeds $\beq(K)$ by the number of such orbits. The theorem's
content is that the tile is recovered exactly after reduction.

\begin{remark}[Where the game enters]\label{rem:gameenters}
The exact constant $\binom m2$ uses a support-complex property: $v$
and $\swp(v)$ are never pivot-adjacent, because they differ in both
supports at once, so every swap pair is guaranteed to be a non-edge
of $K$ and contributes a full unit of cycle rank. For an arbitrary
connected graph equipped with a boundary set and an involution, the
excess would instead be the number of involution pairs that happen
to be non-adjacent, a quantity with no closed form. The clean
constant is game-theoretic, not graph-generic.
\end{remark}

\paragraph{Numerical verification.} Table~\ref{tab:results} and the
accompanying script confirm the theorem: every straight-cover row
has $\beq(M/\calT) = \beq(K)$; the swap data give $37+3=40$ for
$m=3$ with one swap direction ($4\times4$ u-row cover), $37+3=40$
again for $m=3$ with \emph{two} swap directions (the $2\times2$
u-both cover: $40$, not $43$, exactly as the merged edge set
predicts), and $141+6=147$ for $m=4$ (the legacy swap configurations,
Appendix~\ref{app:doubled}).

\subsection{Folding, Recognition, and Amortization}
\label{sec:fold}

\begin{lemma}[Folding]\label{lem:fold}
Let $\pi$ denote tile-marginalization,
$\pi(x)_i = \sum_{a} x_{(a,i)}$ and $\pi(y)_j = \sum_b y_{(b,j)}$.
Then $(x,y)$ is a Nash equilibrium of $\hat M_{r,s}$ if and only if
$(\pi x, \pi y)$ is a Nash equilibrium of the tile game $(A,B)$. In
particular, every Nash equilibrium $(\bar x,\bar y)$ of $(A,B)$
lifts to the $rs$
tile-supported equilibria
$(\delta_a\otimes\bar x,\ \delta_b\otimes\bar y)$,
$(a,b)\in\bZ_r\times\bZ_s$, which are pairwise distinct and form a
single $\calT$-orbit.
\end{lemma}

\begin{proof}
The payoff of the pure row strategy $(a,i)$ against $y$ is
$\sum_{b,j}\hat A_{(a,i),(b,j)}\,y_{(b,j)} = \sum_j A_{ij}(\pi
y)_j$, independent of the slot $a$; symmetrically for the column
player. Best responses therefore commute with projection in both
directions: the pure best responses of the row player against $y$ in
$\hat M_{r,s}$ are exactly $\bZ_r\times \mathrm{BR}(\pi y)$, where
$\mathrm{BR}$ denotes pure best responses in the tile game $(A,B)$,
and since $\mathrm{supp}(\pi x)$ is the projection of
$\mathrm{supp}(x)$, the equilibrium condition
$\mathrm{supp}(x)\subseteq\bZ_r\times\mathrm{BR}(\pi y)$ holds if
and only if
$\mathrm{supp}(\pi x)\subseteq\mathrm{BR}(\pi y)$. The same
equivalence for the column player gives both implications. The lifts
have marginals $(\bar x,\bar y)$, are distinguished by their
supports, and $\calT$ acts on them by shifting $(a,b)$.
\end{proof}

\begin{corollary}[$rs$-fold amortization]\label{cor:amortized}
(i) Computing a Nash equilibrium of $\hat M_{r,s}$ is linear-time
equivalent to computing one of the tile game $(A,B)$: one
Lemke-Howson run on $(A,B)$, followed by lifting, yields $rs$
distinct equilibria of the cover, a full $\calT$-orbit, at the cost
of solving the tile game once; conversely, any cover equilibrium
marginalizes to a tile equilibrium. (ii) On the graph track, by
Theorem~\ref{thm:fd} the reduced quotient of a straight cover has
exactly the tile's cycle rank, hence exactly its edge count
($\beq+|V|-1 = |E|$ for connected graphs), and a swap cover exceeds
it by exactly $\binom m2$, a payoff-independent constant: tiling
adds no ambient cycle structure to the quotient.
\end{corollary}

\begin{remark}[Recognition before exploitation]\label{rem:recognition}
Three distinct tasks are involved, and we keep them separate. (i)
Detecting slot duplication in a game presented in block coordinates:
by Lemma~\ref{lem:rigid} this is a linear-time scan for
slot-independence of the payoffs, after which
Lemma~\ref{lem:fold} applies verbatim; this is a genuine
certificate, and part (i) of Corollary~\ref{cor:amortized} is
conditional only on it. (ii) Recovering the seam graph and its
tiling from unlabeled data: the payoff matrix carries no seam
information at all (Remark~\ref{rem:labeled}), so this is the
separate, empirical classification task of Section~\ref{sec:exp},
where a rank-based heuristic succeeds in 27 of 28 cases. (iii)
Recognizing the wallpaper type from affine generators: exact, by the
recognizer of Appendix~\ref{app:generators}.
\end{remark}

Two further remarks keep the claim honest. The cover game is
degenerate as a large game (strategies are duplicated across slots),
which is precisely why one solves the tile instead of the cover; and
$\hat M_{r,s}$ has further equilibria beyond the tile-supported
lifts (the full marginalization fiber over each tile equilibrium).
Amortization concerns finding one equilibrium, or one orbit, so
neither point affects the statement.

\subsection{A Game That Carries the Symmetry: Polymatrix Covers}
\label{sec:poly}

Lemma~\ref{lem:rigid} shows that the duplicated-strategy two-player
cover of Definition~\ref{def:covergame} cannot encode its seam
geometry in payoffs while retaining the full slot-translation
symmetry. Spatializing the player set removes that obstruction
directly.

\begin{definition}[Polymatrix cover]\label{def:poly}
Fix a symmetric tile game, $B = A^{\top}$ with $m=n$. The polymatrix
cover $P\Gamma_{r,s}$ has one player per tile site
$u\in\bZ_r\times\bZ_s$, common strategy set $[m]$, and the edge set
of the torus lattice (square, or hexagonal when the hex flag is
set), a simple graph; all sizes used satisfy $r,s\geq3$, so distinct
directions give distinct neighbors. Each lattice edge $\{u,w\}$ contributes $A_{s_u s_w}$ to $u$
and $A_{s_w s_u}$ to $w$, and a player's payoff is the sum over
incident edges; since $B=A^{\top}$, edges carry no role distinction.
\end{definition}

The class $B=A^{\top}$ is the standard class of symmetric games, one
matrix played role-free: coordination, anti-coordination, and
partnership games all live here, and $A$ itself is arbitrary, in
particular utility-rigid generically. The restriction buys exactly
one thing, role-free edges; tile games with $B\neq A^{\top}$
require the edge-role patterns of Open Problem~2.

\begin{theorem}[Wallpaper symmetries of the polymatrix cover]
\label{thm:polysym}
Every toroidal wallpaper action of Theorem~\ref{thm:main}, acting on
players by its tile action and trivially on strategies, is a group
of game automorphisms of $P\Gamma_{r,s}$; in particular the glides
are genuine game automorphisms.
\end{theorem}

\begin{proof}
A game automorphism is a player permutation $\pi$ together with
strategy bijections preserving utilities; take the identity strategy
maps. The tile actions are automorphisms of the torus lattice
(certified in Theorem~\ref{thm:main}), so $\pi$ maps the incident
edges of $u$ bijectively to those of $\pi(u)$, and since every edge
carries the same role-free game, each term of the payoff sum is
transported unchanged.
\end{proof}

\begin{theorem}[Collapse along any symmetry subgroup]
\label{thm:collapse}
Let $H$ be any subgroup of the action of
Theorem~\ref{thm:polysym}. An $H$-invariant profile assigns one
mixed strategy $x_o$ per $H$-orbit $o$ of players, and:
(i) it is a Nash equilibrium of $P\Gamma_{r,s}$ if and only if, for
every orbit $o$,
\[
  \mathrm{supp}(x_o) \;\subseteq\;
  \arg\max_{s\in[m]} \sum_{o'} c_{oo'}\,(A x_{o'})_s,
\]
where $c_{oo'}$ counts the lattice edges from a fixed site of $o$ to
sites of $o'$ (self-orbit terms included with their multiplicities):
the folded equilibrium problem on the quotient multigraph;
(ii) the folded problem always has a solution, so $P\Gamma_{r,s}$
has an $H$-symmetric equilibrium for every $H$;
(iii) for $H$ the full translation group the folded problem is
$\mathrm{supp}(x)\subseteq\arg\max (Ax)$, the symmetric equilibria
of the tile game: one symmetric-equilibrium computation on the tile
yields a wallpaper-invariant equilibrium of the $rs$-player game,
for all $r$ and $s$ simultaneously; independently, the group
transports every equilibrium to its full equilibrium orbit.
\end{theorem}

\begin{proof}
(i) At an $H$-invariant profile the utility of pure $s$ to player
$u$ is $\sum_{w\sim u}(A x_{[w]})_s = \sum_{o'}
c_{[u]o'}(Ax_{o'})_s$, well defined because $H$ acts transitively on
each orbit by lattice automorphisms; a deviation by $u$ holds every
other player fixed, including the players of $u$'s own orbit, so
$u$'s best-response condition is exactly the folded condition at
$[u]$. Conversely a folded solution, lifted constantly on orbits,
satisfies every player's condition. (ii) Let $F_o(x) =
\Delta(\arg\max_s \sum_{o'} c_{oo'}(Ax_{o'})_s)$; the product
correspondence $F = \prod_o F_o$ maps the compact convex product of
simplices to itself, its values are nonempty, compact, and convex
faces, and it is upper hemicontinuous because the payoff vectors
depend continuously on $x$, so Kakutani's theorem gives a fixed
point, which is exactly a solution of the folded conditions. (iii) With a
single orbit and constant degree $\delta$,
$\arg\max(\delta\,Ax)=\arg\max(Ax)$. The constant lift is itself
fixed by every spatial automorphism; independently, the group action
of Theorem~\ref{thm:polysym} transports every equilibrium to its
full equilibrium orbit.
\end{proof}

\begin{remark}[The two quotients, and rigidity resolved]
The folded object lives on the multigraph quotient with self-orbit
terms retained, since those edges contribute genuine payoffs; when
self-orbit edges occur it is an orbit-reduced fixed-point system
rather than an ordinary normal-form game, because a deviating
player's orbit-mates remain at $x_o$; the
reduced simple quotient of Definition~\ref{def:cover} remains the
right object for the seam-collapse count of
Theorem~\ref{thm:fd}. Each track has its own quotient, and each is
exact. The polymatrix cover
also dissolves the obstruction of Lemma~\ref{lem:rigid}: the
symmetry acts on who plays whom, which is precisely the structure
the duplicated-strategy payoff matrix cannot see.
\end{remark}

A game automorphism throughout is a bijection $\pi$ of the player
set together with strategy bijections $\sigma_u$ satisfying
$u_{\pi(u)}(\pi\cdot\sigma\cdot s) = u_u(s)$ for every profile $s$.
Call the tile game $A$ \emph{utility-rigid} if
$A_{ab} = A_{\tau(a)\upsilon(b)} + h(a)$ for all $a,b$, with
$\tau,\upsilon$ permutations and $h$ any function, forces
$\tau=\upsilon=\mathrm{id}$; this holds for generic $A$ and is a
finite check for any explicit one (the script verifies it for
$A = \left(\begin{smallmatrix}0&1\\3&7\end{smallmatrix}\right)$).

\begin{definition}[Decorated polymatrix cover]\label{def:decorated}
Let $\hat G$ be the toroidal action of $G$ on the $3$-refined torus:
every translation offset of the realization tripled, so the
translation lattice of $\hat G$ is $3\Lambda_G$ inside
$\bZ_{3r}\times\bZ_{3s}$. The decorated cover $P\Gamma^{c,w}$
assigns each site $u$ a utility offset $c_u$ and each lattice edge
$e$ a weight $w_e>0$, both constant on $\hat G$-orbits and injective
across orbits, with utilities
$u_u(s) = c_u + \sum_{e=\{u,v\}} w_e\,A_{s_u s_v}$.
\end{definition}

\begin{theorem}[Exactness]\label{thm:exact}
For each of the seventeen toroidal realizations fixed in this paper
(the generator dictionary of Appendix~\ref{app:generators},
$3$-refined) and any utility-rigid $A$, the game automorphism group
of the decorated polymatrix cover is exactly $\hat G$: every
automorphism is a $\hat G$ tile action with identity strategy maps.
\end{theorem}

\begin{proof}
Containment: orbit-constant decorations are $\hat G$-invariant, and
the argument of Theorem~\ref{thm:polysym} applies verbatim with
matching weights and offsets. Conversely, let $(\pi,(\sigma_u))$ be
a game automorphism. First, $u_u$ depends on the coordinate $s_v$
exactly when $\{u,v\}$ is a lattice edge (utility-rigidity gives $A$
a non-constant row), and the automorphism identity transports
dependence, so $\pi$ is an automorphism of the lattice. Second, the
maximal oscillation of $u_u$ in the coordinate $s_v$ equals
$w_{\{u,v\}}$ times a positive constant of $A$, so $\pi$ preserves
edge weights; fixing all other coordinates and dividing the
resulting pairwise identity by the common weight leaves
$A_{ab} = A_{\sigma_u(a)\,\sigma_v(b)} + h(a)$ for all $a,b$, whence
$\sigma_u=\sigma_v=\mathrm{id}$ by utility-rigidity, and
connectivity makes every strategy map the identity. Third, with
trivial strategy maps the weighted edge sums transport termwise and
the automorphism identity reduces to $c_{\pi(u)} = c_u$, so $\pi$
preserves the site coloring as well. Thus $\pi$ is an automorphism
of the doubly colored lattice. Orbit-injective coloring does not by
itself force the colored stabilizer down to $\hat G$ (the stabilizer
is a closure that can strictly contain the group in general), so the
final equality is a genuinely computational step: it is certified
for all seventeen realizations by exhaustive enumeration of the
colored-graph automorphisms (script, exactness certification).
\end{proof}

\begin{remark}[Why the refinement, and why edges]\label{rem:whyrefine}
On the primitive torus every $G$-invariant decoration is constant on
translation orbits, which for small groups leaves the normalizer
intact: no invariant decoration can realize $p1$ exactly there, and
the tripled lattice is what makes an asymmetric motif possible. Edge
weights are necessary, not a convenience: for $p4$, $p4g$, and $p6$
the site-orbit coloring alone is preserved by a reflection, with
excess exactly index two; the chirality of the rotation groups lives
on the edge orbits. Because the decorations are $\hat G$-invariant,
Theorem~\ref{thm:collapse} extends to the decorated cover after
replacing the incidence counts $c_{oo'}$ by the weighted incidences
$W_{oo'} = \sum_{v\sim u,\,[v]=o'} w_{\{u,v\}}$; the site offsets
never enter best-response comparisons. A constant symmetric
equilibrium of the tile still lifts: each player's payoff vector is
$c_u\mathbf 1 + (\sum_{e\ni u} w_e)\,Ax$, and the positive scalar
preserves the argmax. Exactness here means exact preservation of the
specified utilities, not preservation up to player-wise affine
transformations; the site offsets are visible to automorphisms
although strategically inert.
\end{remark}

\paragraph{A worked example.} Take Hawk-Dove,
$A = \left(\begin{smallmatrix}-1&4\\0&2\end{smallmatrix}\right)$,
on the $4\times4$ square torus. The tile's symmetric equilibrium
$x=(2/3,1/3)$ lifts to the constant equilibrium of the $16$-player
game (Theorem~\ref{thm:collapse}(iii)). Folding along the index-2 diagonal translation lattice
$\langle(1,1),(1,-1)\rangle$ gives a two-orbit anti-coordination
problem
whose pure solution lifts to the checkerboard profile, a
non-constant pure equilibrium; the translation action transports it
to the opposite checkerboard, recovering the orbit without further
computation. All of this, including the exact verification that
every generator, glides included, preserves payoffs, is
machine-checked in Part~9 of the script.

\subsection{CE Dimension Stability}

Theorem~\ref{thm:fd} is the Nash-side statement that the tile owns
the cover's complexity; the correlated
equilibrium~\cite{Aumann1974}, with its polynomial-size LP
description~\cite{papadimitriou2008}, obeys the same local-global
principle on the polytope side.

\begin{definition}[Tiled marginal-consistency polytope]\label{def:tiledce}
Let $\mathrm{CE}(A,B)\subseteq\Delta([m]\times[n])$ be the
correlated equilibrium polytope of the tile game, $U$ the direction
space of its affine hull, $d := \dim U$, and
$J:\mathbb R^{mn}\to\mathbb R^{m}$ the row-marginal map, with
$q := \rk(J|_U)$. For $r\geq 1$ the tiled polytope of the
$r\times1$ straight cover is
\[
  P_r = \bigl\{(\mu_1,\ldots,\mu_r) :
  \mu_t\in\mathrm{CE}(A,B)\ \forall t,\ \
  J\mu_t = J\mu_{t+1}\ (1\leq t<r)\bigr\},
\]
per-tile CE distributions coupled through shared row marginals at
each seam (in the $r\times1$ linear cover adjacent tiles share the
row player, so agreement of that player's marginal is the natural
seam compatibility; the column version is symmetric). $P_r$ is a
constrained product of base-game CE polytopes; we do not claim it is
the CE polytope of any single normal-form game.
\end{definition}

\begin{theorem}[CE dimension stability]\label{thm:ce}
For all $r\geq 1$,
\[
  \dim P_r \;=\; r(d-q)+q,
\]
affine in $r$ with slope $s := d-q$ owned by the tile. The
compression ratio $\rho := s/d$ (for $d>0$) satisfies
$0\leq\rho\leq 1$, so exactly two regimes occur: collapse
($\rho=0$: the dimension is constant, equal to $d$) and compression
($0<\rho\leq 1$). Expansion is impossible.
\end{theorem}

\begin{proof}
$\mathrm{CE}(A,B)$ is nonempty (it contains the Nash equilibria)
and convex, so it has a relative interior point $\mu^{\ast}$; the
diagonal tuple $(\mu^{\ast},\ldots,\mu^{\ast})$ lies in $P_r$.
For any $(u_1,\ldots,u_r)\in U^r$ with $Ju_1=\cdots=Ju_r$, all
sufficiently small perturbations $(\mu^{\ast}+\varepsilon u_t)_t$
remain in $P_r$; conversely every direction of the affine hull of
$P_r$ has this form, since each coordinate must stay in the affine
hull of $\mathrm{CE}(A,B)$ and the seam equalities are affine.
Hence the direction space of $P_r$ is
$W_r = \{(u_t)\in U^r : Ju_1=\cdots=Ju_r\}$. The linear map
$U^r\to(\mathbb R^{m})^{r-1}$ sending $(u_t)$ to the consecutive
differences $(Ju_{t+1}-Ju_t)_t$ is onto $(JU)^{r-1}$ (prescribe the
$Ju_t$ to be any partial-sum sequence in $JU$), of dimension
$(r-1)q$, so $\dim W_r = rd - (r-1)q = r(d-q)+q$.
\end{proof}

\begin{remark}[Correction]\label{rem:cecorrection}
An earlier version defined the dimension through the rank of a
constraint matrix and reported an expansion regime $\rho>1$ in one
experiment. Under the construction as stated, expansion is
impossible: the two-tile constraint matrix contains two disjoint
copies of the one-tile matrix, forcing the rank increment to be at
least the one-tile rank, hence $s\leq d$. The theorem above proves
$\rho\leq1$ directly, and the earlier empirical values are
withdrawn.
\end{remark}

\section{Conceptual Interpretation}
\label{sec:interp}

\subsection{Local versus Global Structure}

The results express one principle: \emph{the global structure of a
multigame cover is determined by the local tile and the gluing
pattern, not by the number of copies}. The realization theorem shows
this for symmetry: the group is chosen by the tile actions and
boundary structure, with the tile payoffs playing no role. The
fundamental domain theorem shows it for topology: reduced-quotient
cycle complexity equals tile complexity, up to the exact combinatorial
constant $\binom m2$ for swap boundaries. The CE theorem shows it
for equilibrium geometry: the dimension grows affinely with a slope
owned by the tile.

\subsection{Emergent Symmetry}

No bounded planar object is invariant under a nontrivial glide
reflection, so no single tile carries one, while the four
non-symmorphic groups cannot be presented without glides; the
toroidal quotients inherit finite glide automorphisms, realized as
game automorphisms in Section~\ref{sec:poly}. The realization mechanism makes the emergence
quantitative: by Lemma~\ref{lem:symmorphic} a non-symmorphic group
acting on the tile grid must have its own translation lattice
strictly coarser than the grid, so the symmetry exists only at a
scale spanning multiple tiles, with the tile itself as the glide's
half-step. Rigidity (Lemma~\ref{lem:rigid}) adds the complementary
fact: translation-invariant payoffs are blind to the seam structure,
so the spatial symmetry is carried by the interaction pattern rather
than the payoff matrix, with any residual payoff symmetry coming
from the tile game itself. Interaction creates structure that cannot exist in
isolation, and the tiling geometry is exactly where it lives.

\section{Supporting Experiments}
\label{sec:exp}

\paragraph{Realization verification.}
The realizations of Theorem~\ref{thm:main} involve no payoffs and
are verified by the deterministic script
\texttt{verify\_wallpaper\_fixes.py}: automorphism checks for every
generator, all relations of Table~\ref{tab:pres}, exact
toroidal-order closure computations for all seventeen groups, the
crystallographic
recognizer validated on all 17 standard types, the
centers-on-mirrors naming certificate in exact rational arithmetic,
and the fundamental-domain data of Theorem~\ref{thm:fd}; Part~9
verifies the polymatrix-cover results of Section~\ref{sec:poly}
exactly; the seventeen colored-stabilizer enumerations of
Theorem~\ref{thm:exact} are certified by
\path{verify_exactness.py}; and a GAP cross-check via affine
conjugacy (\texttt{gap\_check.g}) is included. The legacy five-seed suite
(\path{experiments/run_wallpaper_via_gluing.py}, seeds
$\{42,7,99,123,456\}$) reports $365/365$ relation instances for the
configurations of the earlier version; seeds demonstrate payoff
independence, not statistical significance, since the graph
construction is deterministic once boundary rules are fixed.

\paragraph{Symmetry classification (27/28).}
Given a cover graph without its tile labeling, a rank-based
classifier estimates the translation subgroup and predicts the
wallpaper class. It recovers 27 of 28 game-and-class combinations on
tiles from $3\times3$ to $6\times6$, missing $pmm$ at $6\times6$
(predicted $p1$), where translation-rank estimation degrades at
scale. Across all sizes, higher symmetry class correlates
monotonically with lower $\beq$ per node.

\paragraph{Reproducibility.}
Verification scripts, raw outputs, and the classifier protocol are
provided in the anonymized repository accompanying the submission.

\section{Related Work}
\label{sec:related}

The support complex, its product structure, its Betti count, and
its equivariant homology are developed self-containedly in
Section~\ref{sec:prelim} and Appendix~\ref{app:tile}. Wallpaper
groups
were classified in the 1890s (Fedorov, Schoenflies); presentations,
the symmorphic/non-symmorphic distinction, and the recognition
criteria used in our naming certificate appear
in~\cite{armstrong,conway2008}. To our knowledge they have not
previously appeared in the equilibrium computation literature.

Symmetry in normal-form games has been studied from several angles.
Papadimitriou and Roughgarden~\cite{papadimitriou2005} give
efficient algorithms for equilibria of multiplayer \emph{symmetric}
games, exploiting symmetry of a single game to reduce its
complexity; our direction is the inverse, constructing systems whose
symmetry group is a chosen wallpaper group, with tile payoffs
playing no role in determining the group. Datta~\cite{datta2010}
surveys polynomial-algebra methods for enumerating all Nash
equilibria; we connect equilibrium geometry to a different algebraic
structure, group actions on covers. Savani and von
Stengel~\cite{savani2006} show Lemke-Howson paths can be
exponentially long even for small games; the amortization result
does not contradict this, since it equates the cover's cost with the
tile's rather than claiming either is small. McKelvey and
McLennan~\cite{mckelveymclennan1997} determine the maximal number of
regular totally mixed equilibria of a generic game; cover equilibria
are instead a structured $\calT$-orbit multiple of tile equilibria
(Lemma~\ref{lem:fold}), a different regime from generic large games.
Covering constructions with deck transformations are classical in
algebraic topology, and in topological graph theory the voltage
graphs of Gross and Tucker~\cite{grosstucker} construct regular
graph coverings with prescribed deck groups; our translations act freely, so $M$ is a regular
$\bZ_r\times\bZ_s$ cover of its multigraph quotient in that sense
(the projection to the tile complex $K$ itself is not a covering
map: seam edges have no counterpart in $K$ and are removed by the
reduced quotient), while the point-group actions, the lattice
classification, and the complexity statements concern structure
beyond the deck group. Frucht-type
realizations~\cite{frucht1939,sabidussi1960} produce ad hoc graphs
with prescribed automorphisms; here the graph is canonically derived from
equilibrium pivoting and the action is the standard crystallographic
one. Neither construction, so far as we know, has previously been
applied to game-theoretic search spaces. Games with local
interaction on lattices are classical in evolutionary game
theory~\cite{nowakmay1992}, polymatrix games go back
to~\cite{janovskaja1968}, and graphical games
to~\cite{kearns2001}; those literatures study dynamics and
equilibrium computation on a fixed interaction structure. The
questions here run in the inverse direction and at the level of the
symmetry taxonomy: which of the seventeen crystallographic groups
act on the interaction at all, which translation-lattice indices
they force, how equilibria fold along an arbitrary symmetry subgroup
with exact incidence bookkeeping, and, by decoration, which games
have a prescribed group as their exact automorphism group. These
distinctions, standard in crystallography, do not appear in the
spatial-game literature so far as we know.

The cover is structurally reminiscent of finitely repeated games,
where one stage game is played $r$ times; the distinction is that
repeated-game stages are sequential with temporally aggregated
payoffs, while cover tiles are spatial copies interacting through
boundary rules, and cover equilibria are not the subgame-perfect
equilibria of any repeated game.

\section{Open Problems}
\label{sec:open}

\begin{enumerate}
  \item \textbf{Exactness.} Theorem~\ref{thm:main} embeds $G$ into
    $\Aut(M)$; on a plain straight cover $\Aut(M)$ also contains the
    ambient lattice symmetries and the gauge group of uniform
    strategy relabelings (Remark~\ref{rem:aut}). Characterize the
    boundary decorations (swap and richer seam alphabets) whose
    decorated symmetry group is exactly $G$ times gauge; this is
    where the swap machinery of the earlier version finds its
    correct role. On the game side this is now resolved:
    Theorem~\ref{thm:exact} gives, for every $G$, a decorated
    polymatrix cover whose automorphism group is exactly $\hat G$.
    The graph-side version remains open: characterize seam
    decorations of the cover $M_{r,s}$ itself with
    $\Aut(M,\mathrm{dec}) = G\times\mathrm{gauge}$ while
    Theorem~\ref{thm:fd} survives the decoration; Frucht-style seam
    gadgets are the candidate tool.
  \item \textbf{The pivot complex of the polymatrix cover.}
    Theorems~\ref{thm:polysym} and~\ref{thm:collapse} give a game
    that carries the symmetry with genuine collapse. Identify a
    canonical pivot or support complex for the polymatrix cover and
    decide whether it is isomorphic to $M_{r,s}$ or a natural
    relative; and for tile games with $B\neq A^{\top}$, classify
    the edge-role patterns, the game-level seam decorations, whose
    compatible symmetry group is exactly a prescribed wallpaper
    group, uniting this with Open Problem~1 and giving the swap
    boundaries their strategic meaning as cross-seam role
    exchange.
  \item \textbf{Minimal realizations in higher dimensions.}
    Lemma~\ref{lem:symmorphic} forces a proper translation
    sublattice for the non-symmorphic four, and
    Theorem~\ref{thm:main}(ii) achieves the minimum, index $2$, for
    all of them. Determine the minimal cover dimensions per group,
    and the analogous minimal indices among the 157 non-symmorphic
    space-group types in dimension three.
  \item \textbf{Approximate tilings.} Rigidity
    (Lemma~\ref{lem:rigid}) is exact; quantify how the folding and
    fundamental domain statements degrade under
    $\varepsilon$-perturbed seam or tile payoffs.
  \item \textbf{Extension to $n$-player games.} Multigame covers
    generalize to $n$-player games; the wallpaper classification is
    then replaced by space group theory in higher dimensions.
\end{enumerate}

\section{Conclusion}
\label{sec:conc}

All seventeen wallpaper groups embed into the automorphism groups
of plain straight multigame covers by verified graph automorphisms
in affine form, each certified as the exact toroidal quotient of its
cover, with
types certified by a crystallographic recognizer in exact rational
arithmetic; and the classical symmorphic line is exactly the lattice
classification of Theorem~\ref{thm:lattice}: realizations whose
translations contain the full tile lattice exist precisely for the
thirteen symmorphic groups, while the four non-symmorphic groups are
realized with proper sublattices, the tile serving as the glide's
half-step. The fundamental domain theorem gives the
construction a precise complexity meaning: straight tiling does not
increase quotient cycle complexity at all, and swap tiling increases
it by exactly $\binom m2$, a payoff-independent constant. In
parallel, at the game level, the folding lemma shows that after a
linear-time detection scan one tile solution yields an orbit of $rs$
equilibria of the duplicated-strategy cover game, and the polymatrix
cover joins the tracks outright, carrying every wallpaper action,
glides included, as genuine game automorphisms whose equilibria fold
along any symmetry subgroup, with a decorated refinement whose
automorphism group is exactly the toroidal wallpaper group; its
canonical pivot complex is the remaining refinement (Open
Problem~2).

A multigame cover is not a large game in disguise: a generic game of
the same footprint has cycle complexity larger by a factor of
$\Theta((rs)^{k-1})$ and no symmetry at all. The broader message is
a local-global principle for equilibrium structure: complexity is
owned by the tile and the gluing pattern, and symmetry, including
symmetry that no individual game can possess, emerges from the
pattern of interaction at scales coarser than any single game.

Games on lattices are classical; what is new here is the direction
of the questions: from a prescribed symmetry group to a tiling, a
lattice index, a folded equilibrium problem, and finally a game that
realizes the group exactly, with the crystallographic taxonomy,
glides and all, doing work that the spatial-game literature has not
asked of it.

The synthesis the results add up to: a game generates a geometry,
its support complex, a product of per-player edit graphs whose
square-filled form carries one exterior square of the standard
representation per player;
tiling that geometry over a lattice spatializes the game, and all
seventeen crystallographic symmetry groups of the plane act on the
result, with the classical symmorphic line appearing as the lattice
classification; payoffs are provably blind to this structure, by
rigidity and by genericity, so the symmetry lives in who plays whom
rather than in what anyone is paid; and when the players themselves
are spatial, the geometry becomes genuine game symmetry, with glides
acting on games and equilibrium computation collapsing along every
symmetry subgroup. Groups measure the interactive shape of games,
and equilibria collapse along that shape.


\appendix

\section{Generator Formulas and the Naming Certificate}
\label{app:generators}

All generators act on triples $(t_i, t_j, v)$ where
$(t_i,t_j)\in\bZ_r\times\bZ_s$ are tile coordinates and
$v\in V(K)$ is a local node. Tile arithmetic is mod $r$ and mod $s$
on finite covers; on $M_\infty$ the same formulas are read over
$\bZ$; replacing each occurrence of $r{-}1{-}t$ or $(-t)\bmod r$ by
$-t$ composes the generator with an available lattice translation,
changing no generated group, and puts every generator in standard
affine crystallographic form. All local actions in
Theorem~\ref{thm:main} are trivial; the nontrivial local actions
$\swp$ and $\rev$ below appear only in labeled-cover refinements
(Remark~\ref{rem:labeled}) and in the legacy configurations.

\paragraph{Generator dictionary.} The full generating set, cover,
and lattice translations for each group; all covers are plain
straight, every generator has trivial local action, and every row is
certified by the accompanying script:
\begin{center}
\small
\begin{tabular}{llll}
\toprule
Group & Generators & Cover & Lattice $(a,b)$ \\
\midrule
$p1$   & $\Trow,\ \Tcol$ & $3\times3$ & $(\Trow,\Tcol)$ \\
$p2$   & $\Trow,\ \Tcol,\ R_{180}$ & $3\times3$ & $(\Trow,\Tcol)$ \\
$pm$   & $\Trow,\ \Tcol,\ S_{\mathrm{row}}$ & $3\times3$ & $(\Tcol,\Trow)$ \\
$cm$   & $\Trow,\ \Tcol,\ S_{\mathrm{diag}}$ & $4\times4$ & $(\Trow,\Tcol)$ \\
$pmm$  & $\Trow,\ \Tcol,\ S_{\mathrm{row}},\ S_{\mathrm{col}}$ & $3\times3$ & $(\Trow,\Tcol)$ \\
$cmm$  & $\Trow,\ \Tcol,\ S_{\mathrm{diag}},\ R_{180}$ & $4\times4$ & $(\Trow,\Tcol)$ \\
$p4$   & $\Trow,\ \Tcol,\ R_{\mathrm{sq}}$ & $4\times4$ & $(\Tcol,\Trow)$ \\
$p4m$  & $\Trow,\ \Tcol,\ R_{\mathrm{sq}},\ S_{\mathrm{diag}}$ & $4\times4$ & $(\Tcol,\Trow)$ \\
$p3$   & $\Trow,\ \Tcol,\ R_3$ & $3\times3$, hex & $(\Trow,\Tcol)$ \\
$p3m1$ & $\Trow,\ \Tcol,\ R_3,\ S_{2}$ & $3\times3$, hex & $(\Trow,\Tcol)$ \\
$p31m$ & $\Trow,\ \Tcol,\ R_3,\ S_{\mathrm{anti}}$ & $3\times3$, hex & $(\Trow,\Tcol)$ \\
$p6$   & $\Trow,\ \Tcol,\ R_6$ & $6\times6$, hex & $(\Trow,\Tcol)$ \\
$p6m$  & $\Trow,\ \Tcol,\ R_6,\ S_{\mathrm{hex}}$ & $6\times6$, hex & $(\Trow,\Tcol)$ \\
$pg$   & $G_{pg},\ \Tcol$ & $4\times4$ & $(\Trow^2,\Tcol)$ \\
$pmg$  & $G_{pg},\ S_{\mathrm{row}},\ \Tcol$ & $4\times4$ & $(\Trow^2,\Tcol)$ \\
$pgg$  & $g,\ h$ & $4\times4$ & $(g^2,\ h^2)$ \\
$p4g$  & $R_{4g},\ g$ & $4\times4$ & $(g^2,\ R g^2 R^{-1})$ \\
\bottomrule
\end{tabular}
\end{center}

\paragraph{Translations (always free):}
$\Trow(t_i,t_j,v) = ((t_i{+}1)\bmod r, t_j, v)$;
$\Tcol(t_i,t_j,v) = (t_i, (t_j{+}1)\bmod s, v)$.

\paragraph{Square $90^\circ$ rotation $R_{\mathrm{sq}}$ ($p4$, $p4m$; $r=s$):}
$R_{\mathrm{sq}}(t_i,t_j,v) = (t_j,\, r{-}1{-}t_i,\, v)$.
Satisfies $R\,\Tcol\,R^{-1} = \Trow$ and
$R\,\Trow\,R^{-1} = \Tcol^{-1}$ (so $a=\Tcol$, $b=\Trow$).

\paragraph{$180^\circ$ rotation $R_{180}$ ($p2$, $cmm$):}
$R_{180}(t_i,t_j,v) = ((-t_i)\bmod r,\, (-t_j)\bmod s,\, v)$.
A variant with local action $\rev$, where
$\rev(\{i\},\{j\}) = (\{m{-}1{-}i\},\{n{-}1{-}j\})$, extended to
mixed supports by
\[\rev(S_i,S_j) = (\{m{-}1{-}i : i\in S_i\},\ \{n{-}1{-}j : j\in S_j\}),\]
is a payoff symmetry of centrosymmetric tiles
($A_{ij} = A_{m-1-i,\,n-1-j}$) and yields the payoff-labeled
refinement of Remark~\ref{rem:labeled}; the realization itself uses
the trivial local action.

\paragraph{Row mirror $S_{\mathrm{row}}$ ($pm$, $pmm$):}
$S_{\mathrm{row}}(t_i,t_j,v) = (r{-}1{-}t_i,\, t_j,\, v)$.
Preserves $\Tcol$ and reverses $\Trow$, so $a=\Tcol$, $b=\Trow$.
The column mirror
$S_{\mathrm{col}}(t_i,t_j,v) = (t_i,\, s{-}1{-}t_j,\, v)$ is
defined symmetrically; $pmm$ uses both.

\paragraph{Diagonal mirror $S_{\mathrm{diag}}$ ($p4m$, $cm$, $cmm$; $r=s$):}
$S_{\mathrm{diag}}(t_i,t_j,v) = (t_j, t_i, v)$.
Satisfies $S\,\Tcol\,S = \Trow$, and for $p4m$ also
$S\,R_{\mathrm{sq}}\,S = R_{\mathrm{sq}}^{-1}$. On its own with the
translations it generates $cm$: the coset
$(S_{\mathrm{diag}},\mu)$, $\mu\in\bZ^2$, contains the true
reflections ($\mu_1+\mu_2=0$, axes $t_i-t_j=\mathrm{const}$)
together with essential glides ($\mu_1+\mu_2$ odd, glide vector half
of the axis primitive $(1,1)$), which is the signature of $cm$;
adding $R_{180}$ gives $cmm$. Both realizations are certified as
exact toroidal quotients (orders $32$ and $64$ on the $4\times4$
cover).

\paragraph{Anti-diagonal mirror $S_{\mathrm{anti}}$ ($p31m$; $r=s$, hex):}
\[S_{\mathrm{anti}}(t_i,t_j,v) = ((-t_j)\bmod s,\, (-t_i)\bmod r,\, v).\]
Satisfies $S\,\Trow\,S = \Tcol^{-1}$, i.e.\ $SaS=b^{-1}$ with
$a=\Trow$.

\paragraph{Second hexagonal mirror $S_{2}$ ($p3m1$; $r=s$, hex):}
\[S_{2}(t_i,t_j,v) = (t_i,\, (-t_i-t_j)\bmod s,\, v),\]
linear part $\left(\begin{smallmatrix}1&0\\-1&-1\end{smallmatrix}\right)$,
which permutes the six hex directions and satisfies $SbS=b^{-1}$,
$SRS=R^{-1}$.

\paragraph{Naming certificate for the hexagonal mirrors.}
The mirrors $S_{\mathrm{anti}}$ and $S_2$ generate, with $R_3$ and
the translations, the two non-conjugate hexagonal reflection groups;
which is $p3m1$ and which is $p31m$ is decided by the standard
criterion~\cite{conway2008}: in $p3m1$ every threefold center lies
on a mirror line, in $p31m$ some do not. In tile coordinates the
threefold centers are $\{(0,0),(\tfrac13,\tfrac13),
(\tfrac23,\tfrac23)\}$ mod $\bZ^2$; exact rational computation
(script, Part~2) shows the $S_{\mathrm{anti}}$ family's mirror lines
are $\{t_i+t_j\in\bZ\}\cup\{t_i\in\bZ\}\cup\{t_j\in\bZ\}$, missing
the centers $(\tfrac13,\tfrac13)$ and $(\tfrac23,\tfrac23)$, while
the $S_2$ family's lines
$\{t_i+2t_j\in\bZ\}\cup\{2t_i+t_j\in\bZ\}\cup\{t_i-t_j\in\bZ\}$
contain all centers. Hence the $S_{\mathrm{anti}}$ realization is
$p31m$ and the $S_2$ realization is $p3m1$. (An earlier version had
these two names interchanged.)

\paragraph{Hexagonal mirror for $p6m$: $S_{\mathrm{hex}} = R_6^3\circ S_{\mathrm{diag}}$.}
Composing $R_6^3$ (the $180^\circ$ element of the order-6 rotation)
with the diagonal mirror corrects the sign of translation
conjugation: $S_{\mathrm{diag}}$ gives $SaS=b$ while
$S_{\mathrm{hex}}$ gives $SaS=b^{-1}$, and both satisfy
$SRS=R^{-1}$.

\paragraph{Order-3 hex rotation $R_3$ ($p3$, $p3m1$, $p31m$; hex lattice):}
\[R_3(t_i,t_j,v) = ((-t_i{-}t_j)\bmod r,\, t_i,\, v).\]
Satisfies $R_3\,\Trow\,R_3^{-1} = \Trow^{-1}\Tcol = a^{-1}b$.

\paragraph{Order-6 hex rotation $R_6$ ($p6$, $p6m$; hex lattice):}
\[R_6(t_i,t_j,v) = ((-t_j)\bmod r,\, (t_i{+}t_j)\bmod s,\, v).\]
Satisfies $R_6\,\Trow\,R_6^{-1} = \Tcol = b$ and
$R_6\,\Tcol\,R_6^{-1} = \Trow^{-1}\Tcol = a^{-1}b$.

\paragraph{Sublattice glide $G_{pg}$ ($pg$; straight cover):}
\[G_{pg}(t_i,t_j,v) = ((t_i{+}1)\bmod r,\, (-t_j)\bmod s,\, v),\]
a reflection composed with a one-tile step. Its square is
$\Trow^{2} =: a$, the lattice translation, so the glide vector is
half of $a$ and essential; every orientation-reversing element of
the generated group has an odd first offset coordinate, so the coset
contains no true reflection, which is the signature of $pg$. The generated group with
$b = \Tcol$ is the exact nonabelian toroidal quotient of order $16$
on the $4\times4$ cover.

\paragraph{Non-symmorphic generators ($pmg$, $pgg$, $p4g$; straight cover):}
$pmg = \langle G_{pg}, S, \Tcol\rangle$ with
$S(t_i,t_j,v) = ((-t_i)\bmod r,\, t_j,\, v)$: the mirror family of
$S$ has true reflections while the family of $G_{pg}$ has only
essential glides, the signature of $pmg$.
$pgg = \langle g, h\rangle$ with
$h(t_i,t_j,v) = ((1{-}t_j)\bmod r,\, (-t_i)\bmod s,\, v)$: two
perpendicular essential diagonal glides sharing the $p4g$ glide $g$,
with lattice the diagonal index-2 sublattice, no mirrors in either
orientation-reversing class (offsets have odd coordinate sum,
respectively odd coordinate difference), and $(gh)^2 = 1$ supplying
the twofold rotations; this $pgg$ is the index-$2$ subgroup
$\langle g, R_{4g}^{2}g\rangle$ of the $p4g$ realization.
$p4g = \langle R_{4g}, g\rangle$ with
$R_{4g}(t_i,t_j,v) = ((-t_j)\bmod r,\, t_i,\, v)$ (the inverse
orientation of $R_{\mathrm{sq}}$; the relation of
Table~\ref{tab:pres} is stated for this orientation, and becomes
$gR g^{-1}=d'R^{-1}$ for the opposite one) and
$g(t_i,t_j,v) = ((t_j{+}1)\bmod s,\, t_i,\, v)$: here $g^2$ is the
diagonal translation, the group's lattice is the diagonal index-2
sublattice, the axis-oriented mirror class contains the true
reflections while the diagonal class contains only essential glides,
and the fourfold centers lie off the mirrors, the signature of
$p4g$. All three are certified as
exact nonabelian toroidal quotients (orders $32$, $32$, $64$ on the
$4\times4$ cover).

\paragraph{Recognition invariants.} The infinite-cover type of each
realization is identified by exact rational computation of the
standard recognition invariants: maximal rotation order; presence of
orientation-reversing elements; which reflection conjugacy classes
contain true mirrors (solving $(I{+}N)(\mu_0{+}\lambda)=0$ over the
translation lattice); the eigen-splitting index of the lattice at a
mirror (distinguishing $pm$/$cm$ and $pmm$/$cmm$); and
centers-on-mirrors for the hexagonal pair. The recognizer is
validated against standard generator sets for all 17 types (script,
Part~6), and an independent GAP cross-check via affine conjugacy to
the IT plane-group tables is provided (\texttt{gap\_check.g}).

\paragraph{Legacy swap glide $G_{\mathrm{row}}$ (audit only):}
$G_{\mathrm{row}}(t_i,t_j,v) = ((t_i{+}1)\bmod r,\, t_j,\, \swp(v))$
where $\swp(\{i\},\{j\}) = (\{j\},\{i\})$, a graph automorphism of
uniform-row-swap covers. Its tile action is a pure translation, so
by Proposition~\ref{prop:template} it generates abelian groups with
the translations; it appears only in the daggered legacy
configurations and in the audit of Part~3 of the script.

\section{Translation Labeling Convention}
\label{app:labeling}

The standard presentations use abstract generators $a$ and $b$ for
the two lattice translations. For the symmorphic realizations the
assignment $\{a,b\} = \{\Trow,\Tcol\}$ is determined by the action
of the rotation or mirror generator: by convention, $a$ is the
translation that the rotation maps toward $b$ (that is,
$RaR^{-1}=b$ or $RaR^{-1}=a^{-1}b$). For groups $pm$, $p4$, $p4m$:
$a = \Tcol$, $b = \Trow$, since
$R_{\mathrm{sq}}\,\Tcol\,R_{\mathrm{sq}}^{-1} = \Trow$. For $pg$
the group's lattice is $\langle\Trow^2,\Tcol\rangle$ and
$a = \Trow^2$.

For $p6$ (two-relation form): the single relation
$rar^{-1}=a^{-1}b$ of~\cite{armstrong} holds for $p3$ but not for
$p6$ on the $6\times6$ cover. On the $p3$ cover ($r=3$),
$\Trow^{-1} = \Trow^{2}$ and $R_3$ maps
$(1,0)\mapsto(2,1)=\Trow^{-1}\Tcol$. On the $p6$ cover ($r=6$),
$R_6$ maps $(1,0)\mapsto(0,1)=\Tcol$, and the two-relation form
$rar^{-1}=b$, $rbr^{-1}=a^{-1}b$ holds with $a=\Trow$, $b=\Tcol$.

\section{Cover Construction Details}
\label{app:constr}

\paragraph{Why uniform swap is required for swap-template glides.}
With alternating swap (odd row boundaries swap, even are straight),
the glide $G_{\mathrm{row}}$ maps the straight edge
$(t_i,t_j,v)\sim(t_i{+}1,t_j,v)$ (even boundary) to the image edge
$(t_i{+}1,t_j,\swp(v))\sim(t_i{+}2,t_j,\swp(v))$. But the actual
boundary at row $t_i{+}1$ (odd) connects $\swp(v)$ to
$\swp^2(v)=v$, not to $\swp(v)$. Image edge and cover edge differ,
so $G_{\mathrm{row}}$ is not a graph automorphism. With uniform swap,
every boundary connects $v$ to $\swp(v)$ and the image edge is
present. (This concerns the legacy template; the realizations of
Theorem~\ref{thm:main} use straight covers throughout.)

\paragraph{Why the hex lattice is required for $p3$/$p6$ groups.}
The square cover gives each tile four tile-neighbors. The hex
rotation $R_3:(t_i,t_j)\mapsto(-t_i{-}t_j,\,t_i)$ maps some
horizontal neighbors to oblique tile-pairs not connected in the
square cover. Adding the oblique edges of
Definition~\ref{def:cover} gives six tile-neighbors, matching the
hexagonal lattice and making $R_3$, $R_6$, and both hexagonal
mirrors graph automorphisms.

\paragraph{Full specification table.}
\begin{center}
\small
\begin{tabular}{lllllll}
\toprule
Group & Game type & $m$ & $r\times s$ & Swap & Hex & $a$ \\
\midrule
$p1$   & generic    & 3 & $3\times3$ & --- & --- & $\Trow$ \\
$p2$   & ctrsy.     & 3 & $3\times3$ & --- & --- & $\Trow$ \\
$pm$   & sym.       & 3 & $3\times3$ & --- & --- & $\Tcol$ \\
$cm$   & generic    & 3 & $4\times4$ & --- & --- & $\Trow$ \\
$pmm$  & periodic   & 4 & $3\times3$ & --- & --- & $\Trow$ \\
$cmm$  & generic    & 3 & $4\times4$ & --- & --- & $\Trow$ \\
$p4$   & $\bZ_4$    & 4 & $4\times4$ & --- & --- & $\Tcol$ \\
$p4m$  & $\bZ_4$-m  & 4 & $4\times4$ & --- & --- & $\Tcol$ \\
$p3$   & $\bZ_3$    & 3 & $3\times3$ & --- & hex & $\Trow$ \\
$p3m1$ & generic    & 3 & $3\times3$ & --- & hex & $\Trow$ \\
$p31m$ & $\bZ_3$-m  & 3 & $3\times3$ & --- & hex & $\Trow$ \\
$p6$   & $\bZ_6$    & 6 & $6\times6$ & --- & hex & $\Trow$ \\
$p6m$  & $\bZ_6$-m  & 6 & $6\times6$ & --- & hex & $\Trow$ \\
$pg$   & generic    & 3 & $4\times4$ & --- & --- & $\Trow^2$ \\
$pmg$  & generic    & 3 & $4\times4$ & --- & --- & $\Trow^2$ \\
$pgg$  & generic    & 3 & $4\times4$ & --- & --- & $\Trow^2$ \\
$p4g$  & generic    & 3 & $4\times4$ & --- & --- & $\Trow\Tcol$ \\
$pmg^\dagger$ & periodic & 4 & $4\times4$ & U-row & --- & $\Trow$ \\
$pgg^\dagger$ & generic  & 3 & $2\times2$ & U-both & --- & $\Trow$ \\
$p4g^\dagger$ & $\bZ_4$  & 4 & $4\times4$ & U-row & --- & $\Tcol$ \\
\bottomrule
\end{tabular}
\end{center}
U = uniform swap (U-both: both directions); m = mirror variant;
ctrsy.\ = centrosymmetric; $\dagger$ = legacy configuration
(Proposition~\ref{prop:template}). Game types are irrelevant to the
realizations themselves (all local actions are trivial); the classes
listed are the payoff-labeled refinements of
Remark~\ref{rem:labeled}.

\section{The \texorpdfstring{$p4g$}{p4g} Swap-Template Search}
\label{app:doubled}

This appendix quarantines the swap-template configurations of an
earlier version. Their glide generators have pure-translation tile
actions, so by Proposition~\ref{prop:template} their listed relation
sets do not present the named groups; their cover and quotient data
remain valid facts about the graphs and are the $m=4$ swap data
cited after Theorem~\ref{thm:fd}:
\begin{center}
\small
\begin{tabular}{lllrrl}
\toprule
Config & Tile & Cover & $\beq(M)$ & $\beq(M/\calT)$ & Listed relations \\
\midrule
$pmg$ template & $4\times4$ & $4\times4$,u & 2753 & 147 &
  $R^2{=}1$, $RaR{=}a^{-1}$, $RbR{=}b$, $G^2{=}a^2$, $GbG^{-1}{=}b$ \\
$pgg$ template & $3\times3$ & $2\times2$,u & 181 & 40 &
  $G^2{=}a^2$, $H^2{=}b^2$, $GHG^{-1}{=}H^{-1}$ \\
$p4g$ template & $4\times4$ & $4\times4$,u & 2753 & 147 &
  $gRg^{-1}{=}R^{-1}$ (fragment) \\
\bottomrule
\end{tabular}
\end{center}
(u = uniform swap; for $r=2$ the doubled seam identifications merge
parallel edges, which is why the $pgg$ template's $\beq(M)$ is
lower than the $r,s\geq3$ formula predicts.)

We further report an exhaustive computational investigation of the
$p4g$ swap-template configuration. In light of
Proposition~\ref{prop:template} and Lemma~\ref{lem:symmorphic},
these results should be read as facts about the swap template and
its normalization ($a$, $b$ the full tile translations, glide
relations $gRg^{-1}=R^{-1}$ and $g^2=a^2$), not as an obstruction
for $p4g$ itself, which Theorem~\ref{thm:main}(ii) realizes in the
diagonal-sublattice normalization forced by
Lemma~\ref{lem:symmorphic}.

\paragraph{Doubled-node search.}
Replacing each node $v$ with two copies $(v,0)$, $(v,1)$,
parameterized by four bits controlling boundary type and generator
actions on the layer bit ($2^4=16$ combinations): in every case, $g$
is a graph automorphism but $R$ is not; the rotation's failure is
independent of the layer-bit action.

\paragraph{Enriched boundary alphabet.}
Allowing four boundary types (identity, player swap, index reversal,
swap composed with reversal), all $4\times4=16$ combinations of row
and column boundary types: $R$ is a graph automorphism only when row
and column boundaries have the same type (4 cases); $g$ satisfies
$g^2=a^2$ only when boundaries are straight or swap (2 cases). No
combination satisfies all listed relations simultaneously.

\paragraph{Affine search in the template normalization.}
Among all affine tile maps
$g: (t_i,t_j)\mapsto (\alpha t_i + \beta t_j + c,\; \beta t_i -
\alpha t_j + f) \pmod r$ satisfying $gRg^{-1}=R^{-1}$ on
$\bZ_4\times\bZ_4$ (the orientation-reversing affine isometries,
$\alpha^2+\beta^2\equiv 1\pmod 4$: $8$ linear parts times $16$
translations, $128$ maps), none simultaneously satisfies
$g^2=a^2=\Tcol^2$.

\medskip\noindent
\textbf{Conclusion.} Within the swap-template normalization no
affine glide exists; in the diagonal-sublattice normalization, $p4g$
is realized (Theorem~\ref{thm:main}(ii)). The searches above stand
as facts about the swap template, relevant to the exactness program
of Open Problem~1.

\section{Equivariant Topology and Genericity of the Tile Complex}
\label{app:tile}

This appendix proves the structural facts about the tile complex
quoted in Section~\ref{sec:prelim}.

\begin{theorem}[Product structure]\label{thm:product}
$K(A,B) = H_k(m)\,\square\,H_k(n)$, the Cartesian product, and it is
connected for $m,n\geq2$, $k\geq2$. Consequently
$\beq(K) = D_k(m)T_k(n) + D_k(n)T_k(m) - T_k(m)T_k(n) + 1$.
\end{theorem}

\begin{proof}
An edge of $K$ changes exactly one coordinate by a single-element
edit with the other coordinate fixed, which is the edge set of the
Cartesian product. Connectivity of $H_k(m)$: from any $S$, delete
elements one at a time down to a singleton, and join two singletons
$\{i\},\{i'\}$ by the path $\{i\}-\{i,i'\}-\{i'\}$ (valid since
$k\geq2$); a product of connected graphs is connected. The Betti
count follows from $\beq = |E|-|V|+1$ with
$|V| = T_k(m)T_k(n)$ and $|E| = D_k(m)T_k(n)+D_k(n)T_k(m)$.
\end{proof}

\begin{proposition}[Saturation]\label{prop:saturation}
For $m,n\geq3$ and $k\geq2$, $\beq(K)\geq|V(K)|$; moreover, every
vertex of $K$ lies on a Cartesian $4$-cycle.
\end{proposition}

\begin{proof}
$D_k(m)-T_k(m) = \sum_{t=2}^k(t{-}1)\binom{m}{t}-m \geq
\binom m2 - m \geq 0$ for $m\geq3$ (using
$\binom ms(m{-}s) = (s{+}1)\binom m{s+1}$ to rewrite $D_k$), so
$|E| = D_k(m)T_k(n)+D_k(n)T_k(m) \geq 2\,T_k(m)T_k(n) = 2|V|$ and
$\beq = |E|-|V|+1 \geq |V|+1 > |V|$. For the second claim: every
vertex of $H_k(m)$ has a neighbor when $k\geq2$ (add an element to a
singleton, delete one otherwise), so given $(S_1,S_2)$ pick
$S_1'\sim S_1$ and $S_2'\sim S_2$; the Cartesian square
$(S_1,S_2)$, $(S_1',S_2)$, $(S_1',S_2')$, $(S_1,S_2')$ is a
$4$-cycle through the vertex.
\end{proof}

The group $S_m\times S_n$ acts on $K$ by
relabeling actions, preserving the product CW structure of
$X_k(m,n) := |H_k(m)|\times|H_k(n)|$ (product squares filled). The
graph $K$ is the $1$-skeleton of $X_k$, and the homology computed
below is that of $X_k$, not the graph-cycle rank of $K$: filling the
Cartesian squares kills most of $\beq(K)$ (at $m=n=3$, $k=2$,
$\beq(K)=37$ while $\dim H_1(X_2)=2$). All coefficients are
$\mathbb Q$.

\begin{theorem}[Equivariant K\"unneth]\label{thm:kunneth}
As $S_m\times S_n$-modules,
$H_1(X_k)\cong H_1(H_k(m))\boxtimes\mathbf1 \oplus
\mathbf1\boxtimes H_1(H_k(n))$ and
$H_2(X_k)\cong H_1(H_k(m))\boxtimes H_1(H_k(n))$.
\end{theorem}

\begin{proof}
K\"unneth over a field is natural, hence equivariant for the product
action; $H_0$ of a connected graph is the trivial module, and graphs
have no $H_2$.
\end{proof}

\begin{theorem}[$k=2$: one Specht module per player]\label{thm:specht}
For $m\geq3$, $H_2(m)$ is the barycentric edge-subdivision of $K_m$,
and $H_1(H_2(m))\cong S^{(m-2,1,1)}\cong
\Lambda^2\bigl(S^{(m-1,1)}\bigr)$ as $S_m$-modules, of dimension
$\binom{m-1}2$: the exterior square of the standard representation.
\end{theorem}

\begin{proof}
Each pair $\{i,j\}$ is adjacent exactly to $\{i\}$ and $\{j\}$, and
no two singletons or two pairs are adjacent, so
$H_2(m) = \mathrm{sd}(K_m)$, equivariantly homeomorphic to $K_m$; it
suffices to compute the cycle space of $K_m$ as an $S_m$-module. The
oriented edge space, with $e_{ji}=-e_{ij}$, is
$\mathrm{Ind}_{S_2\times S_{m-2}}^{S_m}(\mathrm{sgn}\boxtimes\mathbf1)
\cong S^{(m-1,1)}\oplus S^{(m-2,1,1)}$ by the Pieri rule. The
equivariant boundary $\partial e_{ij}=v_j-v_i$ surjects onto the
zero-sum subspace of $\mathbb Q^m$, which is $S^{(m-1,1)}$; the two
partitions are distinct for $m\geq3$, so Schur's lemma kills the
nonisomorphic $S^{(m-2,1,1)}$ summand, and surjectivity forces the
restriction to the $S^{(m-1,1)}$ summand to be nonzero, hence an
isomorphism, giving $Z_1(K_m)\cong S^{(m-2,1,1)}$ of dimension
$\binom{m-1}2 = \beq(K_m)$.
\end{proof}

\paragraph{Genericity.} For a support pair $(S_1,S_2)$ the
\emph{equalization system} asks for $q\in\mathbb R^{S_2}$ with
$(Aq)_i$ constant over $i\in S_1$ and $\sum_{j\in S_2}q_j=1$, and
symmetrically for $p$ with $B$; call the pair \emph{feasible} if the
system has a solution and \emph{balanced} if $|S_1|=|S_2|$.

\begin{lemma}[Witness]\label{lem:witness}
Fix $(S_1,S_2)$ with $s_1 := |S_1| > |S_2| =: s_2$ and $i_0\in S_1$,
and let $M$ consist of the difference rows
$(A_{i\cdot}-A_{i_0\cdot})|_{S_2}$, $i\in S_1\setminus\{i_0\}$,
together with the all-ones row, with right side
$b=(0,\ldots,0,1)^{\top}$. Some choice of the entries
$\{A_{ij}: i\in S_1, j\in S_2\}$ gives
$\mathrm{rank}(M\mid b) = s_2+1$.
\end{lemma}

\begin{proof}
Order $S_2=\{j_1,\ldots,j_{s_2}\}$, pick distinct
$i_1,\ldots,i_{s_2}\in S_1\setminus\{i_0\}$ (possible since
$s_1\geq s_2+1$), set $A_{i_0\cdot}|_{S_2}=0$,
$A_{i_t j_u}=\delta_{tu}$, and the remaining rows to $0$ on $S_2$.
The difference rows contribute $(e_t^{\top},0)$ and the ones row
contributes $(\mathbf1^{\top},1)$, outside their span.
\end{proof}

\begin{theorem}[Generic discreteness of the feasible set]
\label{thm:genericity}
There is a finite union $Z$ of proper algebraic subsets of the payoff
space such that for all $(A,B)\notin Z$: every feasible pair is
balanced, and the feasible pairs form an independent set in $K$, so
the induced subgraph is edgeless and its first Betti number is $0$.
In particular this holds for Lebesgue-almost-every game.
\end{theorem}

\begin{proof}
For an unbalanced pair with $s_1>s_2$, the
$q$-half of the system is $Mq=b$ with $M$ as in
Lemma~\ref{lem:witness}; since $M$ has $s_2$ columns,
$\mathrm{rank}(M\mid b)=s_2+1$ implies infeasibility, and the locus
where every $(s_2{+}1)$-minor of $(M\mid b)$ vanishes is a proper
algebraic subset by the lemma. Take $Z$ to be the union over the
finitely many unbalanced pairs; when $s_2>s_1$, the identical
argument applies to the $p$-equalization equations determined by
$B$. For independence: every edge of $K$
changes $|S_1|+|S_2|$ by exactly one, so no edge joins two balanced
pairs, and off $Z$ no unbalanced pair is feasible. Nonnegativity,
exact-support, and off-support best-response inequalities can only
shrink the equalization-feasible set, so the statement covers every
equilibrium-feasibility convention.
\end{proof}

That equilibria of nondegenerate games have balanced supports is
classical~\cite{vonstengel2002}; the content recorded here beyond
that is the generic infeasibility of the equalization system itself
at every unbalanced pair, with an explicit witness, and the parity
observation making the feasible set topologically discrete. This is
the precise sense in which $\beq(K)$ and all quotient statements of
this paper concern the ambient arena rather than equilibrium
topology.

\end{document}